\documentclass[journal]{IEEEtran}

\usepackage{cite}

  \usepackage{pgfplots}
  \pgfplotsset{compat=newest}
  \usetikzlibrary{plotmarks}
  \usetikzlibrary{arrows.meta}
  \usepgfplotslibrary{patchplots}
  \usepackage{grffile}

\usepackage{tikz}
\usetikzlibrary{arrows.meta,positioning,fit,calc}

\usepackage{graphicx}
\usepackage{xcolor}
\usepackage{tikz}
\usetikzlibrary{shapes,arrows}
\usepackage[table]{xcolor}
\usepackage{booktabs}
\usepackage{balance}

\usepackage[caption=false,font=footnotesize]{subfig} 

\usepackage{amsmath,amssymb,amsfonts, amsthm}

\usepackage[ruled,vlined]{algorithm2e}
\usepackage{enumitem}
\newtheorem{remark}{Remark}
\newtheorem{lemma}{Lemma}
\newtheorem{theorem}{Theorem}

\usepackage{listings}
\begin{document}
\title{On the Structure of the Optimal Detector for Sub-THz Multi-Hop Relays with Unknown Prior: Over-the-Air Diffusion
}

\author{Ozgur Ercetin, and Mohaned Chraiti%
\thanks{The authors are with the Faculty of Engineering and Natural Sciences, Sabanci University, Istanbul, Turkey (e-mail: \{oercetin, mohaned.chraiti\}@sabanciuniv.edu).}%
}

\maketitle

\begin{abstract}
Amplify-and-forward (AF) relaying is an viable strategy to extend the coverage of sub-terahertz (sub-THz) links, but inevitably propagates noise, leading to cumulative degradation across multiple hops. At the receiver, optimal decoding is derisible, yet challenging under non-Gaussian input distributions (video, voice, etc), for which neither the Minimum Mean Square Error (MMSE) estimator nor the mutual information admits a closed form. A further open question is whether or not the knowledge of Channel State Information (CSI) and noise statistics at the intermediate relays are necessary for optimal detection. Aiming for optimal decoder, this paper introduces a new framework that interprets the AF relay chain as a variance-preserving diffusion process and employs denoising diffusion implicit models (DDIMs) for signal recovery. We show that each AF hop is mathematically equivalent to a diffusion step with hop-dependent attenuation and noise injection. Consequently, the entire multi-hop chain collapses to an equivalent Gaussian channel fully described by only three real scalars per block: the cumulative complex gain and the effective noise variance. At the receiver, these end-to-end sufficient statistics define a matched reverse schedule that guides the DDIM-based denoiser, enabling near-optimal Bayesian decoding without per-hop CSI. We establish the information-theoretic foundation of this equivalence, proving that decoding performance depends solely on the final effective Signal-to-Noise-Ratio (SNR), regardless of intermediate noise/channel allocation or prior distribution. Simulations under AWGN and Rician fading confirm that the proposed AF–DDIM decoder significantly reduces mean-squared error, symbol error rate, and bit error rate, particularly at moderate SNRs and for higher-order constellations.

\end{abstract}


\section{Introduction}
Sub-terahertz (sub-THz) communication, operating in the 0.1--1\,THz spectrum, promises ultra-high data rates necessary for future wireless systems~\cite{rappaport2019wireless,koenig2013wireless}. Yet, the pronounced free-space path loss at these frequencies confines viable transmission to short-range scenarios. In indoor environments—comprising corridors, partitions, and reflective surfaces—such attenuation is compounded by complex multipath effects~\cite{goldsmith2005wireless}. To counteract this, fixed-gain amplify-and-forward (AF) relays have emerged as a pragmatic mechanism for extending link coverage across segmented indoor topologies.

While AF relays offer signal amplification, they inherently propagate additive noise, which accumulates across hops and degrades the end-to-end signal integrity. Passive alternatives, such as reconfigurable intelligent surfaces, while immune to noise amplification, lack the power gain required to counterbalance sub-THz path loss. Consequently, fixed-gain AF remains an indispensable architecture in such high-frequency networks, despite its deleterious noise compounding.

At the receiver, optimal detection is fundamentally achieved by the Bayes-optimal minimum mean-square error (MMSE) estimator. However, in AF multi-hop relaying, this detector requires evaluating the posterior distribution of the source given the noisy observations, which entails marginalization over the input prior. In practical scenarios—such as multimedia transmissions of video or voice—the source distribution is generally non-Gaussian, stochastic, and time-varying. Assuming Gaussian priors is therefore unrealistic, and renders closed-form MMSE estimators or mutual information expressions analytically intractable and computationally prohibitive. Consequently, existing works often resort to suboptimal linear MMSE filters, which presume stationary Gaussian noise and fixed end-to-end channel gains, resulting in significant performance loss under realistic AF chains. A further question is whether knowledge of intermediate-hop channel state information (CSI) or noise statistics could improve detection. As we show later, such intermediate information is redundant once the end-to-end sufficient statistics are known.

Motivated by these limitations, this work introduces a new perspective: multi-hop AF relaying can be interpreted as a physical realization of a forward variance-preserving diffusion process \cite{ho2020denoising,sohl2015deep}. Each hop injects Gaussian noise and attenuation, exactly mirroring a diffusion step. At the receiver, the reverse denoising process can be carried out using a schedule derived solely from end-to-end sufficient statistics, without relying on per-hop CSI or processing at intermediate relays. This interpretation aligns with practical constraints, particularly in fixed-gain deployments \cite{bhardwaj2022fixed}, while enabling a fundamentally new class of receiver designs.

Our work advances sub-THz relay networks through:
\begin{itemize}
    \item {Over-the-Air Diffusion:} We establish a direct correspondence between multi-hop AF relaying and a deterministic diffusion forward process, where each hop injects Gaussian noise and attenuation. This interpretation provides a physically grounded generative channel model.

    \item {End-to-End Sufficient Statistics:} We prove that the AF chain collapses into two statistics—a composite complex gain and an effective noise variance—that are information-equivalent to full CSI. Only three real scalars per block are required for SISO, or per eigenmode after whitening in MIMO systems. This equivalence also allows training and reverse denoising with arbitrary numbers of hops, provided the effective SNR is preserved. 

    \item {Channel-Matched Diffusion Scheduling:} We derive a per-sample noise schedule from the equivalent SNR, ensuring that the reverse denoising process is aligned with the realized propagation conditions, without requiring per-hop CSI at the receiver.

    \item {Low-Overhead Implementation:} We quantify the signaling overhead of sufficient statistics, showing that even with conservative quantization, the data utilization exceeds $95\%$ for modest block sizes. This validates our approach as a practical solution for high-throughput sub-THz relay networks.
\end{itemize}


\section{Related Work}
\label{sec:related_work}

\subsection{Sub–THz Links and Multi–Hop Relaying}
Sub–THz communication promises ultra–high data rates but faces severe path loss, atmospheric absorption, and hardware constraints, which limit link budget and range \cite{rappaport2019wireless,koenig2013wireless}. Multi–hop relaying has therefore been widely investigated as a range–extension mechanism \cite{goldsmith2005wireless,boulogeorgos2018terahertz}. Fixed–gain AF is particularly attractive due to low implementation complexity and modulation transparency, and has been advocated for practical deployments \cite{bhardwaj2022fixed}. Performance studies in sub–THz relay settings further quantify the trade–offs among hop count, fading severity, and end–to–end reliability \cite{pai2022performance}. A persistent limitation, however, is the cumulative noise amplification along AF chains, which complicates receiver design beyond two hops.

Closed–form analyses exist for two–hop AF under canonical fading, characterizing end–to–end SNR distributions, outage probability, ergodic capacity, and error rates \cite{qin2019performance}. Short–packet and finite–blocklength regimes for AF have also been studied, highlighting the sensitivity of error exponents to high–SNR approximations and CSI assumptions \cite{gu2017short}. Fixed–gain AF has been examined in over–the–air computation and aggregation scenarios, showing reliability–complexity trade–offs and the benefits of simple relay processing \cite{tang2020reliable}. Most of these models presuppose Gaussian inputs/noise, stationary channels, and limited hop counts; moreover, many rely on per–hop CSI at the receiver. The question of whether \emph{intermediate} CSI/noise statistics are fundamentally necessary for optimal detection in long AF chains remains open in these formulations.

\subsection{Machine Learning for Detection and End–to–End Design}
Learning–based receivers have demonstrated gains over classical linear/MMSE detectors in various settings, including end–to–end autoencoder designs for AF relays \cite{gupta2022endtoend} and deep detectors for random–traffic MIMO \cite{sun2020machine}. Transformer–style detectors have been proposed for cooperative multi–hop channels in related diffusion–based communication media, reporting BER improvements over baselines \cite{xu2021deep}. While promising, most approaches either assume tractable priors or depend on per–hop CSI during training/inference, which limits their applicability in deep AF chains and non–Gaussian traffic scenarios.

\subsection{Diffusion/Score Models in Wireless}
Diffusion probabilistic models and score–based generative methods provide efficient denoisers by learning reverse–time dynamics from corrupted observations \cite{ho2020denoising,sohl2015deep}. Their effectiveness has been established across modalities such as images and speech \cite{dhariwal2021diffusion,chen2020wavegrad}. In wireless contexts, score/diffusion models have been explored for channel estimation and recovery, including Clustered Delay Line (CDL)–based score training for Multi-Input Multi-Output (MIMO) channel posteriors, robust cascaded Ambient Backscatter Communication estimation via annealed sampling, and diffusion priors for ultra–massive MIMO with pilot reduction \cite{arvinte2023score,rezaei2023ambc,zhou2024highdim}. Large–scale channel datasets have also enabled generative modeling workflows for estimation and prediction \cite{alkhateeb2020deepmimo}. These works, however, do not exploit the structural equivalence between AF relaying and variance–preserving forward diffusion, nor do they replace hop–wise CSI with end–to–end sufficient statistics.

\subsection{Positioning of This Work}
The present paper differs from the above literature along three axes.
First, we establish a rigorous AF–to–diffusion equivalence that holds irrespective of the source prior and hop scheduling, proving that any $H$–hop AF chain collapses to a single Gaussian channel whose mutual information and MMSE behavior depend only on the final effective SNR. This result generalizes classical capacity analyses of AF relaying by showing that intermediate statistics are information–theoretically irrelevant.
Second, we prove that two end–to–end sufficient statistics (a composite complex gain and an effective noise variance) are information–equivalent to full per–hop CSI for detection and reverse–schedule construction, enabling signaling with only three scalars per block in the Single-Input Single-Output (SISO) case or per eigenmode after whitening for MIMO.
Third, we leverage this collapse to design a DDIM–based decoder that approximates the Bayesian optimum without requiring explicit priors or intermediate CSI, thereby addressing cumulative noise in deep AF chains under realistic, non–Gaussian sources.

\section{Multi-hop AF Relaying as Diffusion Process}
\label{sec:af_system}

This section specifies the wireless AF relay channel model 
and establishes its fundamental characterization. 
We first present the input-output signal model, including the relay operations, 
channel coefficients, and noise accumulation across hops. 
We then prove that the multi-hop AF relaying process can be mapped to a diffusion process.

\subsection{Signal Model for AF Relaying}
We consider a source $S$, a destination $D$, and a cascade of $H$ AF relays. Transmission proceeds hop-by-hop in synchronized slots within a coherence block over which channels and noise statistics remain constant. Let $X_t\in\mathbb{C}^d$ denote the complex baseband transmit vector at hop $t$ by the relay $R_{t}$ (with $X_0$ produced by the source). The dimension $d$ collects the degrees of freedom (DoF) per block, e.g., $d=1$ (one complex symbol) for Single-Input Single-Output nodes over a channel. An illustrative figure is depicted in Fig. \ref{fig:network}. 

\begin{figure}[!t]
\centering
\includegraphics[width=1\linewidth]{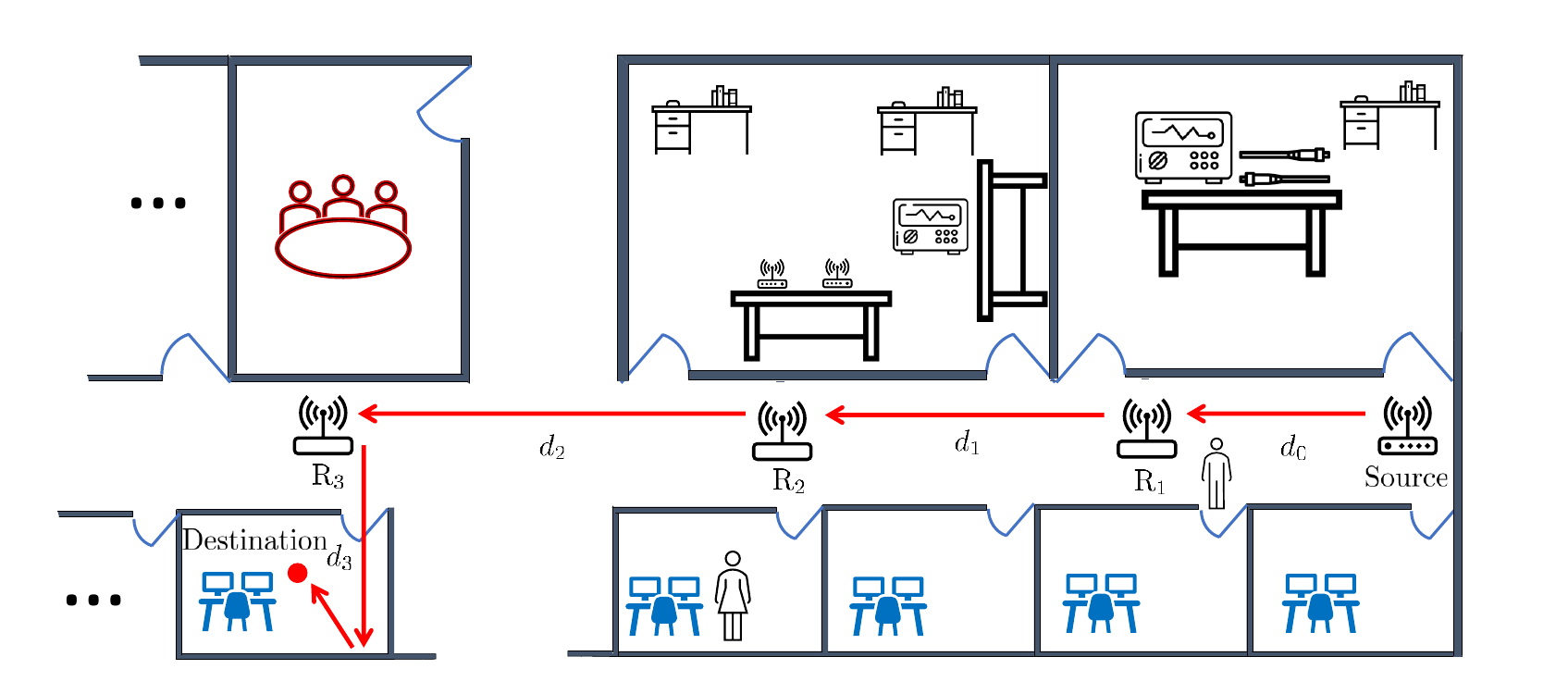}
\caption{Network topology emulating the forward diffusion process in DDPMs.}
\label{fig:network}
\end{figure}

At hop $t \in \{0,\dots,T-1\}$ the relay input vector is
\begin{equation}
    Y_t = \mathcal{H}_t X_t + W_t,
    \label{eq:hop_rx}
\end{equation}
where $\mathcal{H}_t\in\mathbb{C}^{d\times d}$ is the linear channel operator and $W_t$ is zero-mean complex Gaussian noise
\begin{equation}
    W_t \sim \mathcal{CN}(0,\Sigma_t), \qquad \Sigma_t \succ 0.
    \label{eq:noise}
\end{equation}
Each $\mathcal{H}_t$ may include path loss, shadowing, antenna gains, and small-scale fading\footnote{The model accommodates
\textit{i)} Flat fading SISO: $\mathcal{H}_t \in \mathbb{C}$;
    \textit{ii)} Orthogonal Frequency Multiple Access: $\mathcal{H}_t = \mathrm{diag}(h_t[1],\dots,h_t[K])$ after cylic prefix removal/Fast-Fourier-Transform;
    \textit{iii)} Multi-Input-Multi-Output (MIMO)-- OFDMA: block-diagonal $\mathcal{H}_t$ with per-subcarrier MIMO submatrices.}. The noise covariance $\Sigma_t$ captures thermal noise and optional colored components. The second order statistics of the thermal noises $\Sigma_t$ are assumed to be known/estimated.

Relay at $t+1$st hop applies a linear gain $G_t\in\mathbb{C}^{d\times d}$ to its input and transmits
\begin{equation}
    X_{t+1} = G_t Y_t = G_t \mathcal{H}_t X_t + G_t W_t.
    \label{eq:af_tx}
\end{equation}
A per-hop transmit-power constraint is imposed:
\begin{equation}
    \mathbb{E}\|X_{t+1}\|^2 \leq P_{t+1,\max}.
    \label{eq:power_cap}
\end{equation}
On the other hand, we have 
\begin{equation}
\mathbb{E}\|X_{t+1}\|^2 = \mathbb{E}\|G_t \mathcal{H}_t X_t\|^2 +\mathrm{tr}(G_t \Sigma_t G_t^{\!H}),
    \label{eq:af_power_out}
\end{equation}
which must satisfy \eqref{eq:power_cap}.

The per-DoF input SNR of hop $t$ is defined as
\begin{equation}
    \mathrm{SNR}_{\mathrm{in},t} := \frac{ \tfrac{1}{d}\,\mathbb{E}\|\mathcal{H}_t X_t\|^2 }{ \tfrac{1}{d}\,\mathrm{tr}(\Sigma_t) } = \frac{ \mathbb{E}\|\mathcal{H}_t X_t\|^2 }{ \mathrm{tr}(\Sigma_t) }.
    \label{eq:snr_in_def}
\end{equation}
For flat-fading SISO, $\mathrm{SNR}_{\mathrm{in},t}=|h_t|^2 P_t/\sigma_t^2$ with $\Sigma_t=\sigma_t^2$. 

In the case of colored noise where $\Sigma_t\neq \sigma_t^2 I$(a problem that does not occur for the case of SISO system), one can apply a whitening transform:
\begin{align}
    \widetilde{Y}_t &= \Sigma_t^{-1/2} Y_t, \quad \widetilde{H}_t=\Sigma_t^{-1/2} \mathcal{H}_t,\\ \quad \widetilde{W}_t&=\Sigma_t^{-1/2} W_t \sim \mathcal{CN}(0,I_d).
    \label{eq:whiten}
\end{align}
Hence, 
\begin{align}
 \widetilde{X}_{t+1} &= G_t \widetilde{H}_t X_t + G_t \widetilde{W}_t,
    \label{eq:af_compact}
\end{align}
Therefore, for the egneral case, whether noise is white or colored, each relay multiply the received signal by $\Sigma_t^{-1/2}$ as such the result noise is of second order statistics $I_d$, i.e.,  $\widetilde{W}_t\sim\mathcal{CN}(0,I_d)$. For notation simplity, in the rest of the paper we use uniform notation $X_{t}$, $W_t$ and $\mathcal{H}_t$ to denote $\widetilde{X}_t$, $\widetilde{W}_t$ and $\widetilde{W}_t$. 

\subsection{AF Relaying as a Variance-Preserving Diffusion Step}
\label{Sec:ARdiffusion}

We now demonstrate that a multi-hop AF chain realizes exactly the same recursion as the variance-preserving (VP) diffusion model. 
In VP diffusion, the corrupted states evolve according to
\begin{equation}
    X_{t+1} = \alpha_t X_t + \sqrt{1-\alpha_t^2}\,Z_t,
    \qquad Z_t \sim \mathcal{CN}(0,I_d),
    \label{eq:vp_af_equiv}
\end{equation}
where $\alpha_t^2 = 1-\beta_t$ and $\beta_t \in (0,1)$ denotes the diffusion strength at step $t$.

Consider the AF relay model in \eqref{eq:hop_rx}–\eqref{eq:af_tx} with whitening applied so that $W_t \sim \mathcal{CN}(0,I_d)$. The transmit signal at hop $t+1$ is
\begin{equation}
    X_{t+1} = G_t \mathcal{H}_t X_t + G_t W_t.
    \label{eq:af_vp1}
\end{equation}
Let $\mu_t^2=\mathbb{E}\|X_t\|^2$ and define normalized states $\tilde{X}_t=X_t/\mu_t$, which ensures $\mathbb{E}\|\tilde{X}_t\|^2=1$. 
Similarly, denote $\mu_{t+1}^2=\mathbb{E}\|X_{t+1}\|^2$ and write $\tilde{X}_{t+1}=X_{t+1}/\mu_{t+1}$. 
Substituting into \eqref{eq:af_vp1} gives
\begin{equation}
    \tilde{X}_{t+1} = \frac{G_t \mathcal{H}_t \mu_t}{\mu_{t+1}}\,\tilde{X}_t + \frac{G_t}{\mu_{t+1}}\,W_t.
    \label{eq:af_vp2}
\end{equation}

The output power is
\begin{equation}
    \mu_{t+1}^2 = \mathbb{E}\|X_{t+1}\|^2 
    = \|G_t \mathcal{H}_t\|_F^2\,\mu_t^2 + \mathrm{tr}(G_t G_t^{\!H}).
    \label{eq:af_power}
\end{equation}
given that $W_t\sim\mathcal{CN}(0,I_d)$. Here, $\|\cdot\|_F$ denotes the Frobenius norm, i.e.,
$
\|\cdot\|_F = \sqrt{\mathrm{tr}(\cdot (\cdot)^{\!H})}.
$ From \eqref{eq:af_vp2}, the effective coefficient of $\tilde{X}_t$ is therefore
\begin{equation}
    \alpha_t = \frac{\|G_t \mathcal{H}_t\|\mu_t}{\mu_{t+1}},
\end{equation}
and the variance of the normalized noise is
\begin{equation}
    \frac{\|G_t\|_F^2}{\mu_{t+1}^2}.
\end{equation}
Combining with \eqref{eq:af_power}, we obtain
\begin{align}
    \alpha_t^2 
    &= \frac{\|G_t \mathcal{H}_t\|_F^2 \mu_t^2}{\|G_t \mathcal{H}_t\|_F^2 \mu_t^2 + \|G_t\|_F^2} 
     = \frac{\mathrm{SNR}_{\mathrm{in},t}}{1+\mathrm{SNR}_{\mathrm{in},t}}, \label{eq:alpha_snr}
\end{align}
where the per-hop input SNR is defined as
\begin{equation}
    \mathrm{SNR}_{\mathrm{in},t} 
    = \frac{\mathbb{E}\|\mathcal{H}_t X_t\|^2}{\mathrm{tr}(\Sigma_t)}.
\end{equation}
Substituting \eqref{eq:alpha_snr} back into the noise variance expression shows that
\begin{equation}
    \frac{\|G_t\|_F^2}{\mu_{t+1}^2} = 1-\alpha_t^2,
\end{equation}
which identifies the variance of the injected Gaussian noise with the diffusion parameter $\beta_t=1-\alpha_t^2$. 

Thus the normalized recursion \eqref{eq:af_vp2} takes the exact VP form
\begin{equation}
    \tilde{X}_{t+1} = \alpha_t \tilde{X}_t + \sqrt{1-\alpha_t^2}\,Z_t,
    \qquad Z_t\sim\mathcal{CN}(0,I_d).
\end{equation}
To conclude, each AF hop is mathematically identical to a VP diffusion step, with diffusion strength $\beta_t$ determined solely by the input SNR. A cascade of AF relays therefore implements the forward VP diffusion process.

\section{DDPM and the information flow: only the final SNR matters}
While the previous section establishes the analogy between AF relaying and the forward process of DDPM, it is important to recall that DDPM originates from a generative modeling perspective, whereas AF relaying arises in a communications setting. 
To connect these two viewpoints rigorously, an information-theoretic analysis of the diffusion recursion is required. The objective of this section is to show that the diffusion and reverse processes can be determined solely by final effective SNR: the same amount of information is preserved regardless of number of hops and how attenuation/noise are distributed across intermediate hops. Once the cumulative attenuation and effective SNR are known, the detailed channel gains and noise statistics at intermediate relays are immaterial to the fundamental detection limits. An intuitive illustration is as follows: a chain of $100$ hops is 
equivalent—in terms of end-to-end information transfer—to a chain of $10$ hops with noise fraction as long as they yields to the same SNR. As long as both 
scenarios yield the same effective SNR at the destination, they yields the same information-theoretic limits.  

This equivalence has two critical implications for learning and for reverse diffusion. First, the reverse process 
depends only on the effective SNR, not on the actual number of hops or the details noise addition over each of the hops. Consequently, training and reverse processes can be performed with a 
different number of steps (normally smaller) than the number of hops, as long as the cumulative SNR matches. This means that training across a range of SNR levels is sufficient to cover all possible AF relay configurations. Regardless of how the channel varies from hop to hop and over time, only the final SNR matters for efficient reversal of the process and for achieving optimal detection. Second, 
it reduces system identification to a simple task: estimating the end-to-end SNR is sufficient to construct the 
reverse diffusion process.

\subsection{Collapsed Gaussian Channel Representation of Forward Recursions}

Consider a generic noisy recursion that models repeated attenuation and Gaussian perturbation. 
Let $\{Z_t\}_{t\geq 0}$ be i.i.d.\ standard Gaussian vectors in $\mathbb{C}^d$, independent of $X_0$. 
For coefficients $\alpha_t\in[0,1)$, define
\begin{equation}
   X_{t+1} = \alpha_t X_t + \sqrt{1-\alpha_t^2}\,Z_t, \qquad t=0,1,\dots,T-1.
   \label{eq:17}
\end{equation}
Here $\alpha_t$ denotes the attenuation parameter, while $\beta_t=1-\alpha_t^2$ quantifies the injected noise power. 
The cumulative attenuation and the corresponding effective SNR after $H$ steps are given by
\begin{equation}
   \bar\alpha_H = \prod_{t=0}^{T-1}\alpha_t, 
   \qquad \gamma_H = \frac{\bar\alpha_H^2}{1-\bar\alpha_H^2}.
\end{equation}
Although the recursion in \eqref{eq:17} appears to generate a complex sequence of random variables, the following lemma shows that the entire chain collapses into a single Gaussian channel.

\begin{lemma}[Collapsed form]\label{lem:collapsed}
The $H$-step recursion reduces to
\begin{equation}
    X_H = \bar\alpha_H X_0 + \sqrt{1-\bar\alpha_H^2}\,Z_0,
\end{equation}
where $Z_0\sim \mathcal{N}(0,I_d)$ is independent of $X_0$.
\end{lemma}

\begin{proof}
The claim is immediate for $H=1$. Assume it holds for $H$ and write
\begin{align*}
    X_{T+1} &= \alpha_H X_H + \sqrt{1-\alpha_H^2}\,Z_H \\
            &= \alpha_H \bar\alpha_H X_0 
             + \alpha_H\sqrt{1-\bar\alpha_H^2}\,Z_0
             + \sqrt{1-\alpha_H^2}\,Z_H.
\end{align*}
The last two terms are independent Gaussian vectors with total variance 
$\alpha_H^2(1-\bar\alpha_H^2)+(1-\alpha_H^2)=1-\bar\alpha_{T+1}^2$, 
and independence from $X_0$ is preserved. 
\end{proof}

Lemma~\ref{lem:collapsed} shows that, regardless of the number of steps or how noise is accumulated, the recursion always admits a \emph{single equivalent Gaussian channel} description: the input $X_0$ undergoes fading by $\bar\alpha_H$ and independent Gaussian perturbation of variance $1-\bar\alpha_H^2$. 

While the equivalent channel description captures the end-to-end form of the recursion, it does not by itself determine whether knowledge of the intermediate distributions affects the amount of information preserved. The following theorem resolves this by showing that the mutual information between $X_0$ and $X_H$ depends only on the effective SNR and is independent of the details of the intermediate steps.

\begin{theorem}[Mutual information collapse]\label{thm:MIcollapse}
Let $X_0$ be an arbitrary random vector in $\mathbb{C}^d$ with finite second moment. 
For the recursion in \eqref{eq:17}, the mutual information between $X_0$ and the $H$-step output $X_H$ is given by
\begin{equation}
    I(X_0;X_H) = J(\gamma_H),
\end{equation}
where $\gamma_H=\frac{\bar\alpha_H^2}{1-\bar\alpha_H^2}$ is the effective SNR, and $J(\cdot)$ is a continuous, strictly increasing function on $[0,\infty)$ for every non-degenerate $X_0$. 
In particular, for small SNR,
\begin{equation}
    I(X_0;X_H) = \frac{\gamma_H}{2}\,\mathbb{E}\|X_0\|^2 + o(\gamma_H), 
    \qquad \gamma_H \downarrow 0.
\end{equation}
Hence, the mutual information depends only on the end-to-end SNR $\gamma_H$ and is completely insensitive to the allocation of attenuation and noise across intermediate steps.
\end{theorem}

\begin{proof}
By Lemma~\ref{lem:collapsed}, the recursion reduces to
\[
X_H \stackrel{d}{=} \bar\alpha_H X_0 + \sqrt{1-\bar\alpha_H^2}\,Z_0,
\]
with $Z_0\sim\mathcal{N}(0,I_d)$ independent of $X_0$. 
Equivalently,
\[
X_H \stackrel{d}{=} \sqrt{\tfrac{\gamma_H}{1+\gamma_H}}\,X_0 + \tfrac{1}{\sqrt{1+\gamma_H}}\,Z_0,
\]
which is exactly the output of a real/complex Gaussian channel $Y_\gamma=\sqrt{\gamma}\,X_0+Z_0$ with $\gamma=\gamma_H$. 
Therefore,
\[
I(X_0;X_H)=I(X_0;Y_{\gamma_H}).
\]

The I--MMSE identity of Guo--Shamai--Verdú \cite{Guo2005I_MMSE} states that
\[
\frac{d}{d\gamma} I(X_0;Y_\gamma) = \tfrac{1}{2}\,\mathrm{mmse}_{X_0}(\gamma),
\]
where $\mathrm{mmse}_{X_0}(\gamma)=\mathbb{E}\!\big[\|X_0-\mathbb{E}[X_0\mid Y_\gamma]\|^2\big]$ is the minimum mean-square error in estimating $X_0$ from $Y_\gamma$. 
Since $\mathrm{mmse}_{X_0}(\gamma)>0$ for all $\gamma>0$ when $X_0$ is non-degenerate, the function $\gamma\mapsto I(X_0;Y_\gamma)$ is strictly increasing and continuous. 
Furthermore, expanding around $\gamma=0$ yields
\[
I(X_0;Y_\gamma)=\frac{\gamma}{2}\,\mathbb{E}\|X_0\|^2+o(\gamma),
\]
which gives the stated asymptotic behavior.

The key observation is that $\gamma_H$ depends on the product $\prod_{t=0}^{T-1}\alpha_t$ only through the scalar $\bar\alpha_H^2$, and thus two different recursion schedules $\{\alpha_t\}$ and $\{\alpha'_t\}$ that yield the same $\bar\alpha_H$ (equivalently the same $\gamma_H$) produce the same mutual information. 
Hence, the entire information-theoretic behavior of the recursion is captured by the end-to-end SNR $\gamma_H$, and the precise distribution of attenuation and noise across intermediate steps is immaterial.
\end{proof}

\noindent\textbf{Discussion:} 
Theorem~\ref{thm:MIcollapse} establishes that the mutual information depends solely on the end-to-end SNR $\gamma_H$. 
This is a strong conclusion: the detailed schedule $\{\alpha_t\}$ and the manner in which noise is accumulated across steps are irrelevant. 
From the standpoint of detection limits, only the final effective SNR matters. 
In the context of AF relaying, this means that a long chain of many noisy hops is indistinguishable from a shorter chain with appropriately matched SNR. 
Consequently, the reverse diffusion process can be trained and executed independently of the specific AF setup and the instantaneous channel state information. In particular, it may employ fewer and coarser steps than the actual number of relays, provided that the cumulative SNR is preserved. This observation eliminates the need to retrain the DDPM and relaxes the requirement for knowledge of intermediate relay states. 

\subsection{Optimal Power Allocation and Scheduling}

Although the end-to-end information transfer depends only on the final effective SNR, 
the distribution of power across relays determines how this SNR is achieved. 
Different power allocation strategies can therefore influence both the efficiency 
of resource usage and the stability of the reverse diffusion process, even under 
a fixed total relay power budget. We next characterize the optimal allocation.

\begin{theorem}[Optimal Power Allocation Across Relays]\label{thm:equalization}
Consider the recursion with $\alpha_t^2 = 1-\beta_t$ and
\[
\beta_t = \frac{1}{1+\mathrm{SNR}_{\mathrm{in},t}}
= \frac{1}{1+c_t P_t}, 
\]
where $P_t\in[0,P_{t,\max}]$ is the transmit power of relay $t$, 
and $c_t=\frac{|h_t|^2}{\sigma_t^2}>0$ is the channel-to-noise ratio parameter at hop $t$. 
Then
\[
\bar\alpha_H^2=\prod_{t=0}^{T-1}(1-\beta_t)
=\prod_{t=0}^{T-1}\frac{c_tP_t}{1+c_tP_t}.
\]

Given a total power budget $P_{\mathrm{tot}}>0$, the optimal allocation is obtained by solving
\begin{align}\label{eq:opt}
&\max_{\{P_t\}}\ \sum_{t=0}^{T-1}\log\frac{c_tP_t}{1+c_tP_t}\\
&\quad \text{s.t.}\quad
\sum_{t=0}^{T-1} P_t \le P_{\mathrm{tot}},\qquad 
0 \le P_t \le P_{t,\max}.
\nonumber
\end{align}
This optimization is strictly concave and admits a unique solution $\{P_t^\star\}$. 
For each interior variable $0<P_t^\star<P_{t,\max}$, the Karush--Kuhn--Tucker (KKT) 
conditions reduce to the equal-marginal rule
\[
\frac{1}{P_t^\star(1+c_tP_t^\star)}=\mu,
\]
for some multiplier $\mu>0$ chosen so that the total budget is met. 
\end{theorem}

\begin{proof}
Define
\[
g_t(P) = \log\frac{c_tP}{1+c_tP}, \qquad P>0,\ c_t>0.
\]
We first verify concavity. The first and second derivatives are
\[
g_t'(P) = \frac{1}{P(1+c_tP)} > 0, \qquad
g_t''(P) = -\frac{1+2c_tP}{P^2(1+c_tP)^2} < 0.
\]
Hence $g_t(P)$ is strictly increasing and strictly concave on $(0,\infty)$. 
It follows that the sum $\sum_t g_t(P_t)$ is strictly concave and that 
the feasible set in \eqref{eq:opt} is convex. Therefore, the problem admits 
a unique optimizer.

Next, consider the Lagrangian
\[
\begin{aligned}
\mathcal{L} &= \sum_{t=0}^{T-1} g_t(P_t) 
- \mu\Bigl(\sum_{t=0}^{T-1} P_t - P_{\mathrm{tot}}\Bigr)\\
&\qquad+ \sum_{t=0}^{T-1} \nu_t P_t 
+ \sum_{t=0}^{T-1} \omega_t (P_{t,\max}-P_t),
\end{aligned}
\]
where $\mu\ge0$ is the multiplier for the total power constraint, 
and $\nu_t,\omega_t\ge0$ are multipliers for the box constraints.

The KKT stationarity condition yields
\[
g_t'(P_t^\star) = \mu - \nu_t + \omega_t.
\]
For any interior solution $0<P_t^\star<P_{t,\max}$, the complementary slackness 
conditions imply $\nu_t=\omega_t=0$, and therefore
\[
g_t'(P_t^\star) = \mu.
\]
Substituting the explicit form of $g_t'(P)$ gives
\[
\frac{1}{P_t^\star(1+c_tP_t^\star)} = \mu,
\]
which can be rearranged as
\[
P_t^\star(1+c_tP_t^\star) = \frac{1}{\mu}.
\]

This quadratic equation in $P_t^\star$ has the positive solution
\[
P_t^\star = \frac{\sqrt{1+\tfrac{4c_t}{\mu}}-1}{2c_t}.
\]
Boundary cases $P_t^\star=0$ or $P_t^\star=P_{t,\max}$ follow from the 
corresponding slackness conditions. Finally, since $g_t'(P)>0$, the 
budget constraint is tight unless prevented by the upper bounds 
$P_{t,\max}$. 
\end{proof}

\begin{remark}[Equalization Principle]
If $c_t\equiv c$ and no upper bounds are active, then 
$P_t^\star=P_{\mathrm{tot}}/T$ for all $t$, so that 
$s_t^\star=cP_t^\star$ and $\beta_t^\star=(1+s_t^\star)^{-1}$ 
are equalized across hops. 
In general, however, the optimizer does not equalize $\beta_t$; 
it equalizes the marginal gain $g_t'(P_t)=1/[P_t(1+c_tP_t)]$ 
across active relays.
\end{remark}

The optimization in Theorem~\ref{thm:equalization} provides the \emph{power allocation rule across hops}. 
While the collapsed mutual information depends only on $\bar\alpha_H$, the \emph{learnability and stability} of the reverse process depend on how the attenuation is distributed. 
This is directly relevant for AF relay networks: equalizing the marginal benefit across hops ensures efficient use of total power and leads to smoother reverse denoisers, improving practical recovery performance.

\subsection{Feasibility of Reliable Communication under Arbitrary Priors}
\label{sec:feasibility}

The AF--VP equivalence shows that after $H$ hops the received state is
\begin{equation}
    X_H = \bar\alpha_H X_0 + \sqrt{1-\bar\alpha_H^2}\,Z,
    \qquad Z\sim \mathcal{CN}(0,I_d),
    \label{eq:collapsed_siso}
\end{equation}
where $\bar\alpha_H=\prod_{t=0}^{T-1}\alpha_t$. 
Thus the end-to-end channel is equivalent to a memoryless Gaussian channel with effective SNR
\begin{equation}
    \mathrm{SNR}_{\mathrm{eff}}
    = \frac{\bar\alpha_H^2 \, \mathbb{E}\|X_0\|^2}{1-\bar\alpha_H^2}.
    \label{eq:eff_snr}
\end{equation}
The feasibility of reliable communication is therefore determined solely by $\mathrm{SNR}_{\mathrm{eff}}$ and the prior distribution of $X_0$.

\subsubsection{Gaussian and Discrete Priors}

If $X_0 \sim \mathcal{CN}(0,P I_d)$, mutual information admits the closed form
\begin{equation}
    I(X_0;X_H) = \log \det \!\big(I_d + \mathrm{SNR}_{\mathrm{eff}} I_d\big),
    \label{eq:gaussian_capacity}
\end{equation}
so the feasibility condition is $R \le I(X_0;X_H)$.  
If this bound is violated, no decoder can achieve reliable communication.  

For discrete constellations such as $M$ Quadratic Amplitude Modulation (QAM), standard symbol error rate approximations take the form
\begin{equation}
    \mathrm{SER} \approx f_M(\mathrm{SNR}_{\mathrm{eff}}),
\end{equation}
but these formulas assume uniform symbol priors. In practice, $P(X_0)$ may be non-uniform due to coding, probabilistic shaping, or machine-learned mappings, which invalidates closed-form SER thresholds.

\subsubsection{General Information-Theoretic Characterization}

For arbitrary priors $P_{X_0}$ with finite second moment, feasibility is still fully characterized by the mutual information
\begin{equation}
    I(X_0;X_H) = \mathbb{E}\!\left[
        \log \frac{p(X_H \mid X_0)}{p(X_H)}
    \right],
\end{equation}
where $p(X_H \mid X_0)$ is Gaussian with variance $1-\bar\alpha_H^2$.  

The I--MMSE identity \cite{Guo2005I_MMSE} provides a rigorous characterization:
\begin{equation}
    \frac{d}{d\gamma} I(X_0;Y_\gamma)
    = \tfrac{1}{2}\,\mathrm{mmse}_{X_0}(\gamma),
    \qquad Y_\gamma = \sqrt{\gamma}\,X_0+Z,
    \label{eq:immse}
\end{equation}
with
\begin{equation}
    \mathrm{mmse}_{X_0}(\gamma)
    = \mathbb{E}\!\Big[\,
        \|X_0-\mathbb{E}[X_0 \mid Y_\gamma]\|^2
    \,\Big].
\end{equation}
Thus,
\begin{equation}
    I(X_0;X_H) = \tfrac{1}{2}\int_0^{\mathrm{SNR}_{\mathrm{eff}}} 
    \mathrm{mmse}_{X_0}(\gamma)\, d\gamma.
\end{equation}
Eq~\eqref{eq:immse} shows that feasibility is governed entirely by the prior-dependent MMSE curve. 
Closed forms exist only for Gaussian priors; nevertheless, the limits are uniquely determined by $\mathrm{SNR}_{\mathrm{eff}}$ for any distribution.

\subsection{Implications for AF--Diffusion Detection}

Classical feasibility thresholds (e.g., QAM SER formulas) are thus special cases of the general I--MMSE framework under restrictive assumptions.  
In realistic scenarios, the prior $P_{X_0}$ is structured, non-uniform, or unknown, making closed-form thresholds inapplicable.  

In this regime, the optimal detector is the Bayesian estimator
\begin{equation}
    \hat{X}_0 = \mathbb{E}[X_0 \mid X_H],
\end{equation}
which requires exact knowledge of $P_{X_0}$. Diffusion-based generative models (DDPM, DDIM) provide a universal approximation mechanism: by learning the score functions $\nabla \log p_t(x)$ of the corrupted states, the reverse process implements near-optimal Bayesian denoising without requiring explicit priors.  
Hence, diffusion models act as universal decoders: for Gaussian priors they reduce to the MMSE detector in closed form, and for arbitrary or unknown priors they approximate the Bayesian optimum, bypassing the need for prior-specific feasibility thresholds.

\section{Implementation and Numerical Results}
\subsection{Implementation and Practical Considerations}
\label{sec:implementation}

The AF--VP equivalence has direct impact on receiver design. 
The collapsed channel after $H$ hops is 
\[
X_H = \mu_H X_0 + \sqrt{v_H}\,Z,\qquad Z\sim\mathcal{CN}(0,I),
\]
where $(\mu_H,v_H)$ are the cumulative gain and noise variance. 
This description is sufficient for both information-theoretic analysis 
and practical reverse denoising.

\subsubsection{Overhead}

Instead of forwarding full per-hop channel state information, only three real scalars 
$(\Re\{\mu_H\},\Im\{\mu_H\},v_H)$ are needed per block. 
This makes the overhead negligible: for an $N\times N$ QAM block with 
$k=\log_2 M$ bits/symbol, the utilization factor is
\[
\eta=\frac{N^2 k}{N^2 k + B_{\text{CSI}}},\qquad 
B_{\text{CSI}}=b_{\Re\mu}+b_{\Im\mu}+b_v.
\]
With moderate quantization ($B_{\text{CSI}}=80$--$96$ bits), 
$\eta\ge0.95$ once $N\ge 20$ for $M=16$ or $N\ge 12$ for $M=256$. 
Thus end-to-end sufficiency yields large savings relative to hop-wise signaling.

\begin{algorithm}[t]
\caption{AF--DDIM Detection}
\label{alg:af_ddim}
\KwIn{Received block $X_H$, sufficient statistics $(\mu_H, v_H)$, reverse steps $H$, trained denoiser $\varepsilon_\theta(\cdot)$.}
\KwOut{Estimate $\widehat X_0$.}

\textbf{Equalize:} $x_H \leftarrow X_H / \mu_H$.  \tcp*{$\tilde X_H$}

\textbf{Schedule from end-to-end stats:}
\;
$\alpha_{\mathrm{bar}}^2 \leftarrow \dfrac{|\mu_H|^2}{|\mu_H|^2 + v_H}$; \quad
$\alpha_{\mathrm{bar}} \leftarrow \sqrt{\alpha_{\mathrm{bar}}^2}$; \\
$\alpha_{\mathrm{step}} \leftarrow \alpha_{\mathrm{bar}}^{1/T}$; \quad
$\beta_{\mathrm{step}} \leftarrow 1 - \alpha_{\mathrm{step}}^2$; \\
\For{$t \leftarrow 0$ \KwTo $H$}{
  $\bar\alpha_t \leftarrow \alpha_{\mathrm{step}}^t$; \quad
  $\sigma_t \leftarrow \sqrt{1 - \bar\alpha_t}$; \quad
  $\lambda_t \leftarrow \log\!\left(\dfrac{\bar\alpha_t^2}{1 - \bar\alpha_t^2}\right)$ \tcp*{log-SNR (optional embedding)}
}

\textbf{Reverse denoising (DDIM, deterministic):}\\
\SetKw{KwDownTo}{downto}
\For{$t \leftarrow T$ \KwDownTo $1$}{
  $e_t \leftarrow \varepsilon_\theta\!\big(x_t,\ t,\ \lambda_t,\ \mu_H,\ v_H\big)$ \tcp*{$\varepsilon$-prediction network}
  $x^{(t)}_0 \leftarrow \dfrac{x_t - \sigma_t\, e_t}{\sqrt{\bar\alpha_t}}$ \tcp*{predicted clean at step $t$}
  $x_{t-1} \leftarrow \sqrt{\bar\alpha_{t-1}}\, x^{(t)}_0 \;+\;
     \sqrt{1 - \bar\alpha_{t-1}}\;\dfrac{x_t - \sqrt{\bar\alpha_t}\,x^{(t)}_0}{\sigma_t}$ \;
}
\Return{$\widehat X_0 \leftarrow x_0$}
\end{algorithm}

\subsubsection{Reverse Process from End-to-End SNR}

From $(\mu_H,v_H)$, the effective SNR is
\[
\mathrm{SNR}_{\mathrm{eq}}=\frac{|\mu_H|^2}{v_H}, 
\qquad 
\bar\alpha^2=\frac{\mathrm{SNR}_{\mathrm{eq}}}{1+\mathrm{SNR}_{\mathrm{eq}}}.
\]
A matched reverse schedule is then constructed as
\[
\bar\alpha_t=\bar\alpha^{\,t/T},\quad
\sigma_t=\sqrt{1-\bar\alpha_t},
\]
for $t=0,\ldots,T$, with $H$ reverse iterations. 
This removes the need for per-hop CSI: the reverse recursion depends only 
on $(\mu_H,v_H)$ and the step index.

Equalization by $\mu_H$ produces $\tilde X_H=X_H/\mu_H$, which initializes the reverse chain. 
The DDIM-style recursion is
\begin{align}
x_0^{(t)} &= \frac{x_t - \sigma_t \,\widehat\varepsilon_t}{\sqrt{\bar\alpha_t}},\\
x_{t-1} &= \sqrt{\bar\alpha_{t-1}}\,x_0^{(t)} 
+ \sqrt{1-\bar\alpha_{t-1}}\,
\frac{x_t-\sqrt{\bar\alpha_t}\,x_0^{(t)}}{\sigma_t},
\end{align}
with $\widehat\varepsilon_t=\varepsilon_\theta(x_t,t,\lambda_t,\mu_H,v_H)$, 
and $\lambda_t=\log(\bar\alpha_t^2/(1-\bar\alpha_t^2))$ as log-SNR embedding.

\subsubsection{Learning-Based Denoiser}

The denoiser $\varepsilon_\theta$ approximates the Bayesian MMSE estimator 
$\hat X_0=\mathbb{E}[X_0\mid X_H]$ by predicting the injected Gaussian noise. 
Training data are generated by propagating constellation symbols through 
random AF chains, recording $(X_H,\mu_H,v_H)$ across diverse hop counts, 
fading realizations, and power allocations. 
The loss is the standard diffusion objective:
\[
\mathcal{L}(\theta)=\mathbb{E}\left\|\varepsilon_\theta(X_H,t,\mu_H,v_H)-Z\right\|^2,
\]
where $Z$ is the known Gaussian corruption. 
Since the AF chain collapses to $(\mu_H,v_H)$, 
the denoiser requires no access to intermediate CSI. 

In our implementation, the denoiser operates on I/Q blocks as two-channel images using a compact 2D U-Net. 
At each reverse step, it receives the state $x_t \in \mathbb{R}^{2\times N\times N}$, the step index $t$, and a scalar noise-level $\ell_t$ (e.g., $\ell_t=\log\!\tfrac{\bar\alpha_t}{1-\bar\alpha_t}$). 
A sinusoidal embedding of $t$ and an MLP embedding of $\ell_t$ are fused to modulate the residual blocks via FiLM-style additive biases. 
The U-Net uses two downsampling and two upsampling stages with base width $C_0{=}64$, GroupNorm, SiLU activations, and skip connections. 
The network outputs the predicted noise $\varepsilon_\theta(x_t,t,\ell_t)$, which is used in the DDIM update to form $x_{t-1}$.




\begin{figure*}[!t]
    \centering
    \subfloat[MSE]{%
        \includegraphics[width=0.32\textwidth]{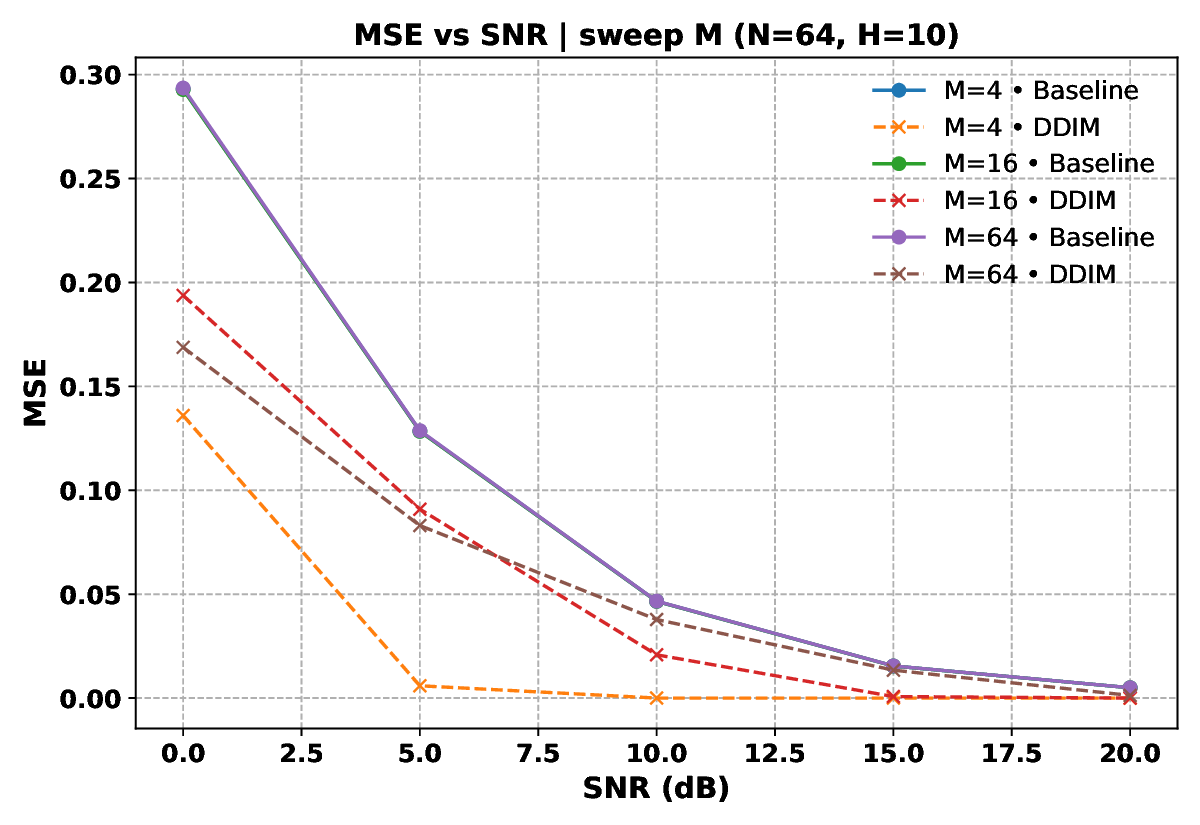}%
        \label{fig:mse_n64_h10}}
    \hfill
    \subfloat[SER]{%
        \includegraphics[width=0.32\textwidth]{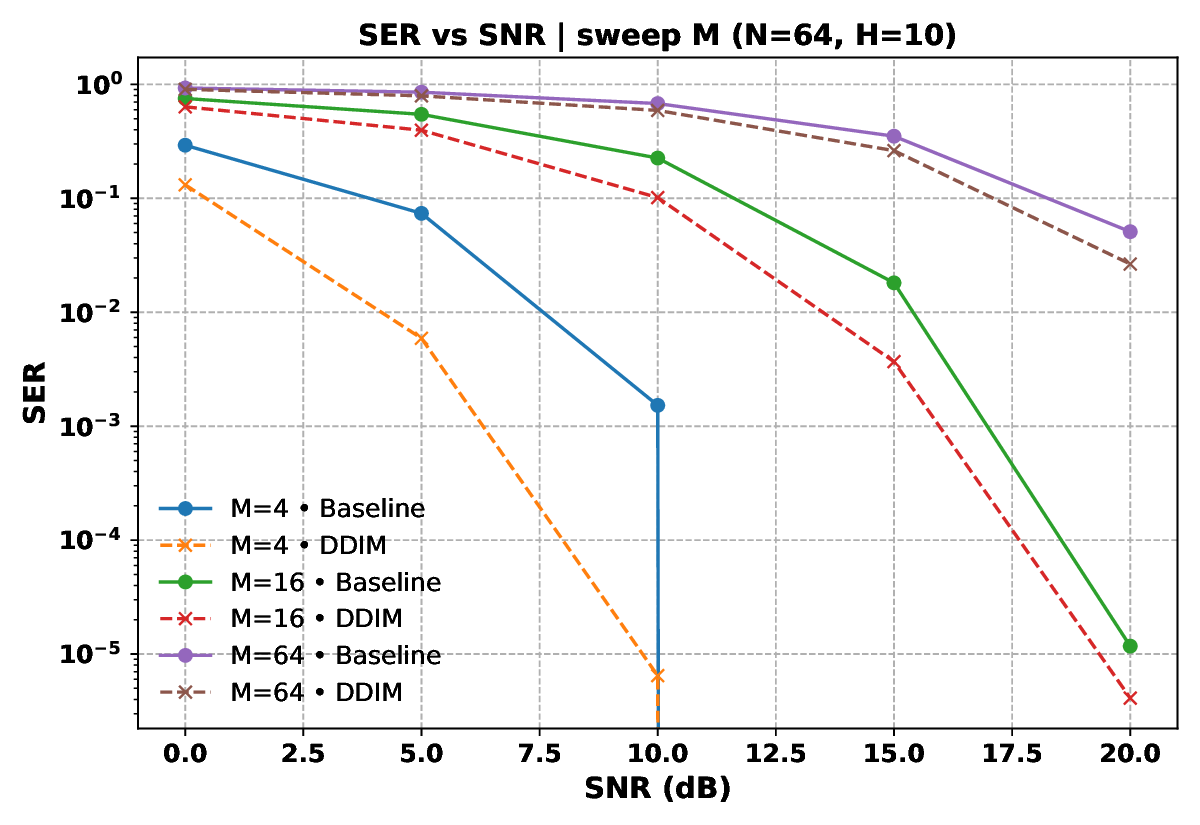}%
        \label{fig:ser_n64_h10}}
    \hfill
    \subfloat[BER]{%
        \includegraphics[width=0.32\textwidth]{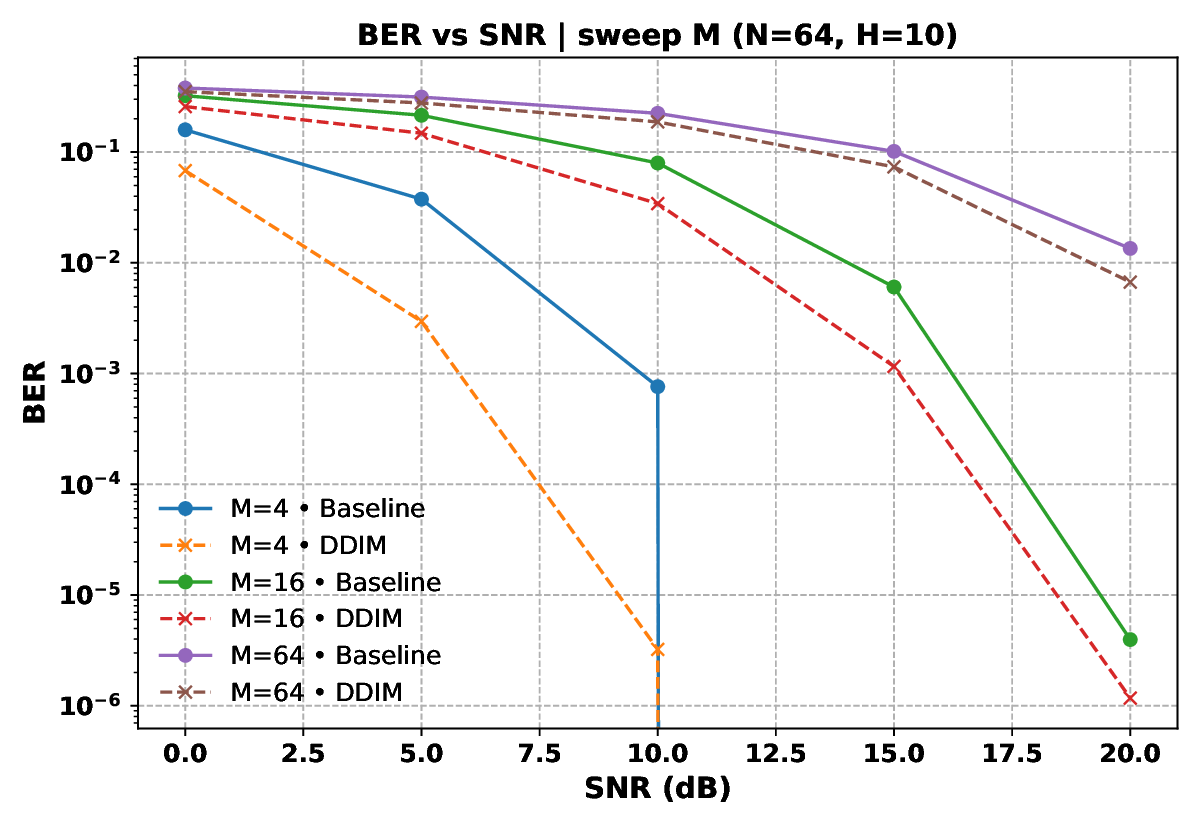}%
        \label{fig:ber_n64_h10}}
    \caption{Performance under AWGN-only propagation (no fading) for a block size of $N=64$ and $H=10$ hops. 
Curves compare baseline direct detection (no denoising) with AF--DDIM reconstruction for modulation orders $M\in\{4,16,64\}$: 
(a) mean-squared error (MSE), (b) symbol error rate (SER), and (c) bit error rate (BER). 
Results show that AF--DDIM consistently reduces reconstruction error and improves detection accuracy, with pronounced gains for higher-order constellations and moderate SNR ranges.}
    \label{fig:perf_vs_M_n64_h10}
\end{figure*}

Figure~\ref{fig:perf_vs_N_m16_h10} evaluates the effect of coherence block size $N$ under AWGN-only propagation with $M=16$ and $H=10$ hops. 
In Fig.~\ref{fig:mse_m16_h10}, larger $N$ improves performance across the SNR range, reflecting greater symbol diversity per block and enhanced stability of the matched diffusion schedule. 
The symbol error rate (Fig.~\ref{fig:ser_m16_h10}) and bit error rate (Fig.~\ref{fig:ber_m16_h10}) likewise improve with increasing $N$. 
These trends confirm that the AF--DDIM framework remains effective even for modest block sizes under AWGN noise accumulation.

\begin{figure*}[!t]
    \centering
    \subfloat[MSE]{%
        \includegraphics[width=0.32\textwidth]{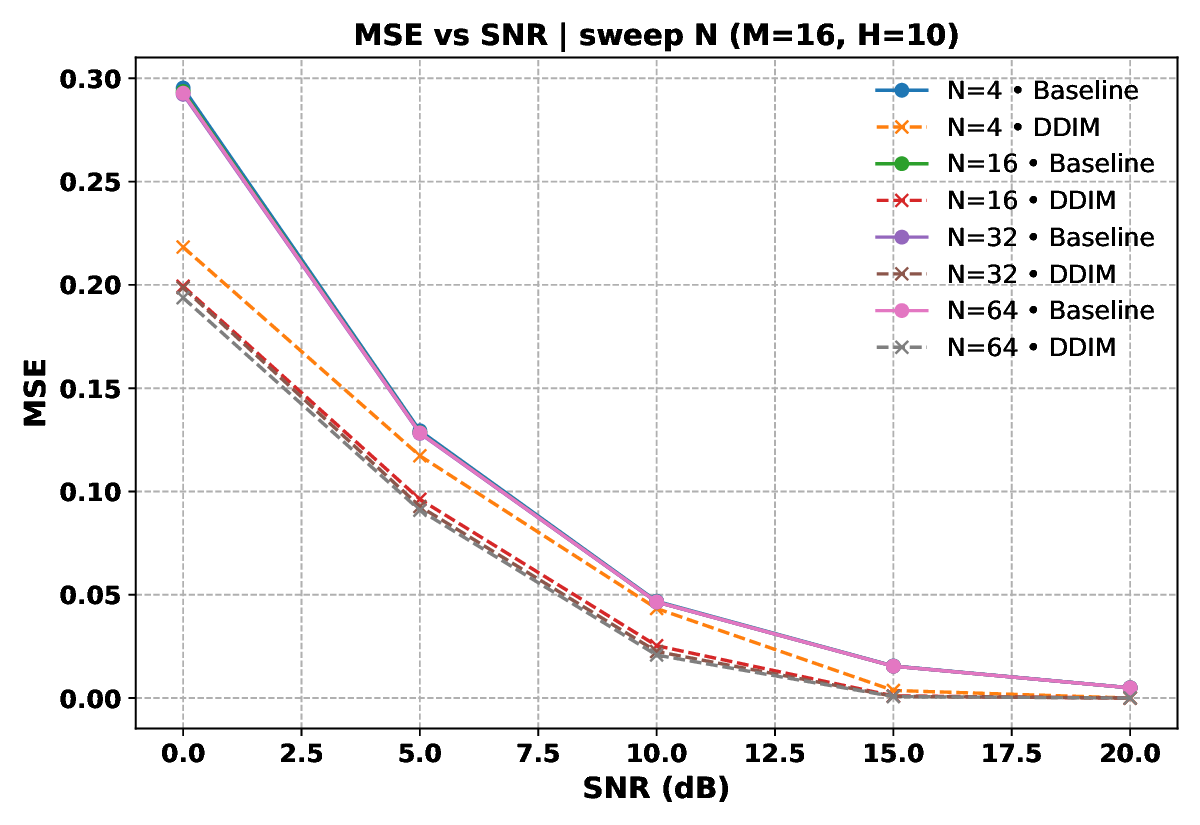}%
        \label{fig:mse_m16_h10}}
    \hfill
    \subfloat[SER]{%
        \includegraphics[width=0.32\textwidth]{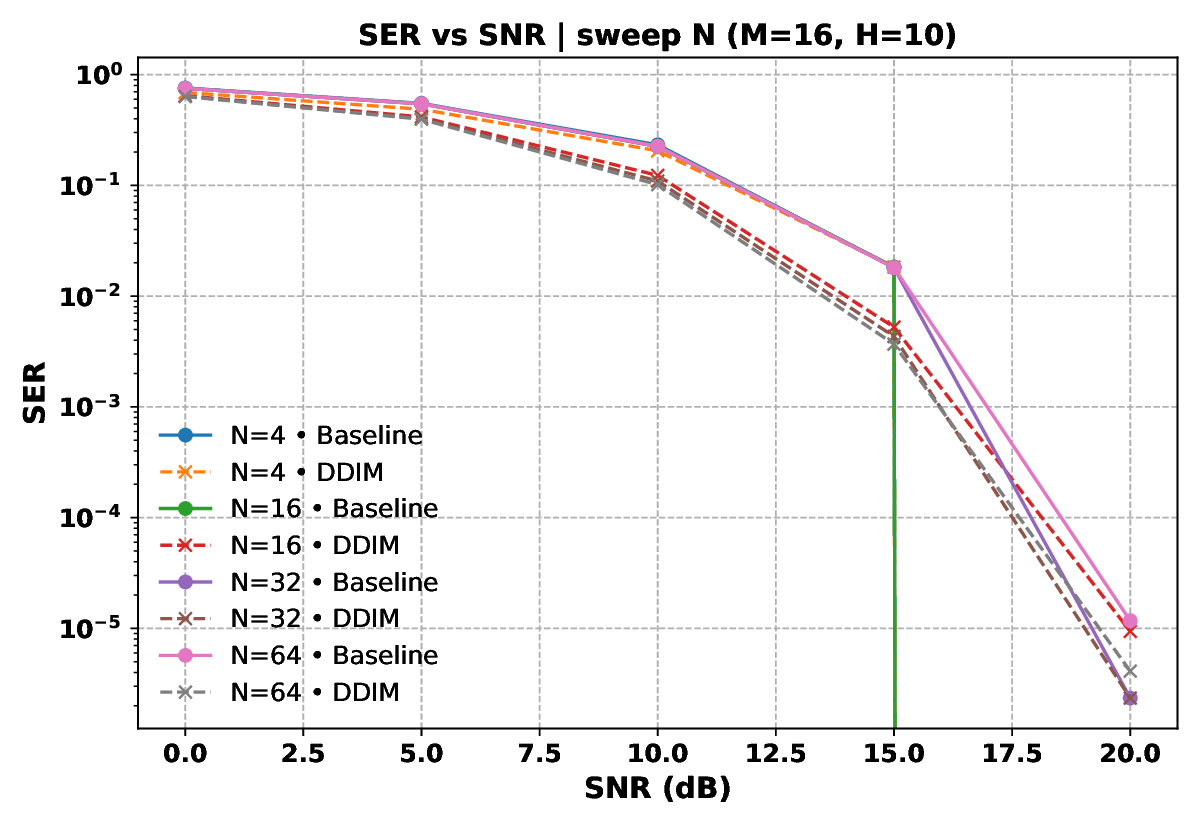}%
        \label{fig:ser_m16_h10}}
    \hfill
    \subfloat[BER]{%
        \includegraphics[width=0.32\textwidth]{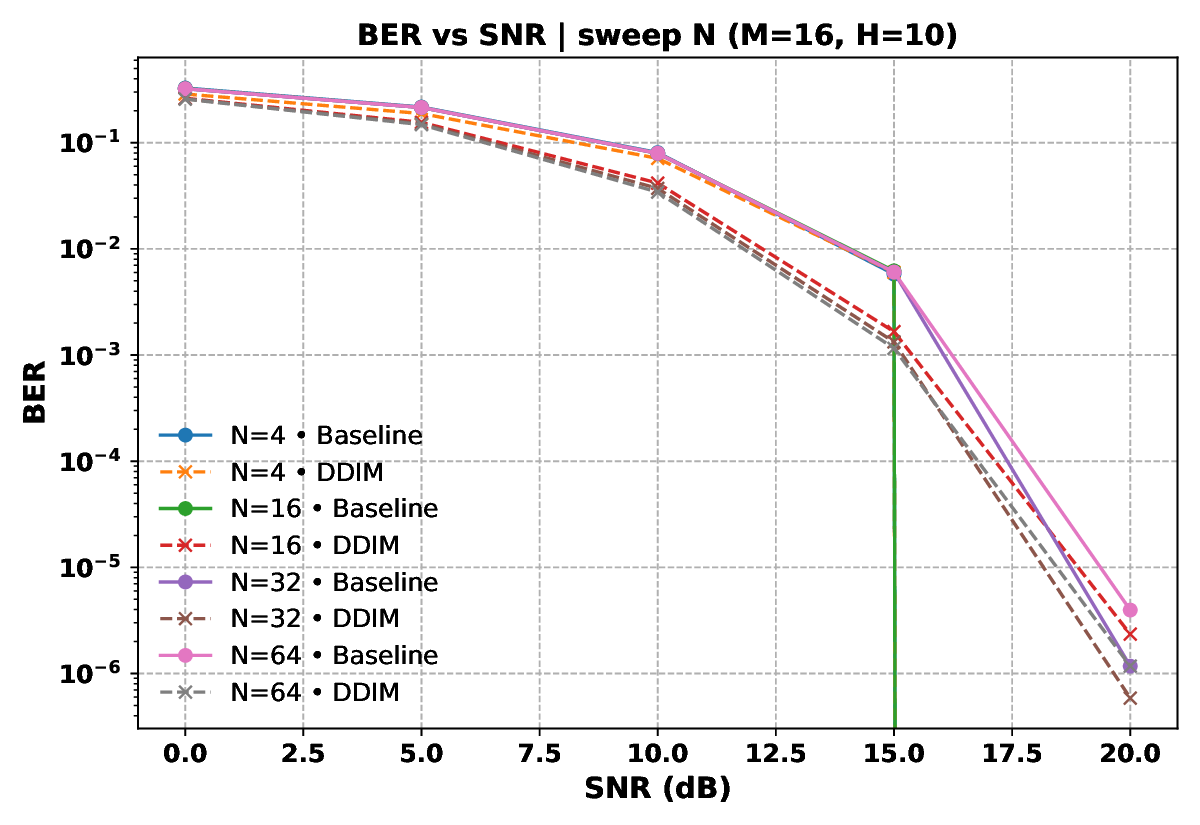}%
        \label{fig:ber_m16_h10}}
    \caption{Performance under AWGN-only propagation for $M=16$ and $H=10$ hops, as a function of coherence block length $N$: 
(a) MSE, (b) SER, and (c) BER. 
Results show that increasing $N$ improves detection accuracy by amortizing the constant CSI overhead, with AF--DDIM maintaining clear performance gains across metrics.}
    \label{fig:perf_vs_N_m16_h10}
\end{figure*}

Figure~\ref{fig:perf_vs_H_n64_m16} examines the effect of hop count $H$ under AWGN-only propagation for $M=16$ and block size $N=64$. 
As $H$ increases, the accumulated relay noise grows, leading to higher error rates in baseline detection. 
In contrast, AF--DDIM reconstruction remains stable: in Fig.~\ref{fig:mse_n64_m16}, the MSE penalty with increasing $H$ is significantly mitigated, 
and in Figs.~\ref{fig:ser_n64_m16}--\ref{fig:ber_n64_m16}, both SER and BER remain consistently lower than baseline across all $H$. 
These results demonstrate the robustness of the proposed framework to multi-hop noise accumulation.

\begin{figure*}[!t]
    \centering
    \subfloat[MSE]{%
        \includegraphics[width=0.32\textwidth]{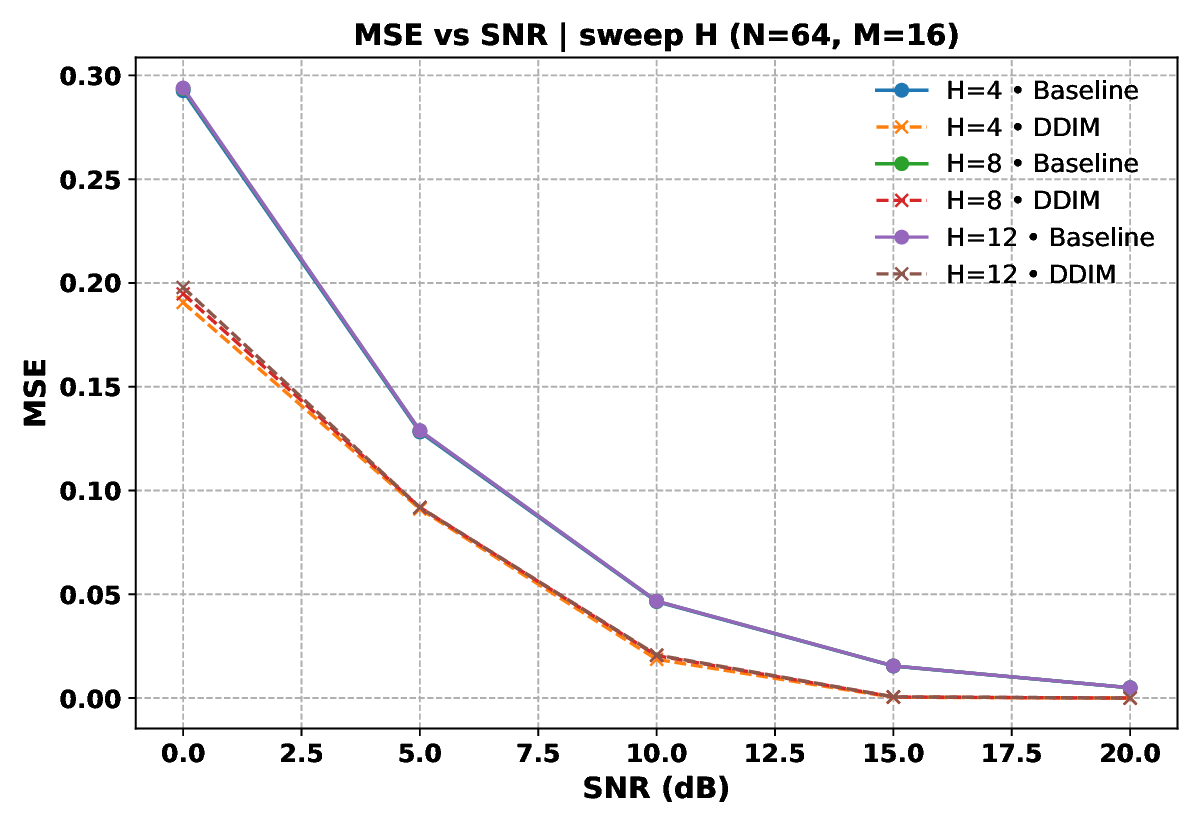}%
        \label{fig:mse_n64_m16}}
    \hfill
    \subfloat[SER]{%
        \includegraphics[width=0.32\textwidth]{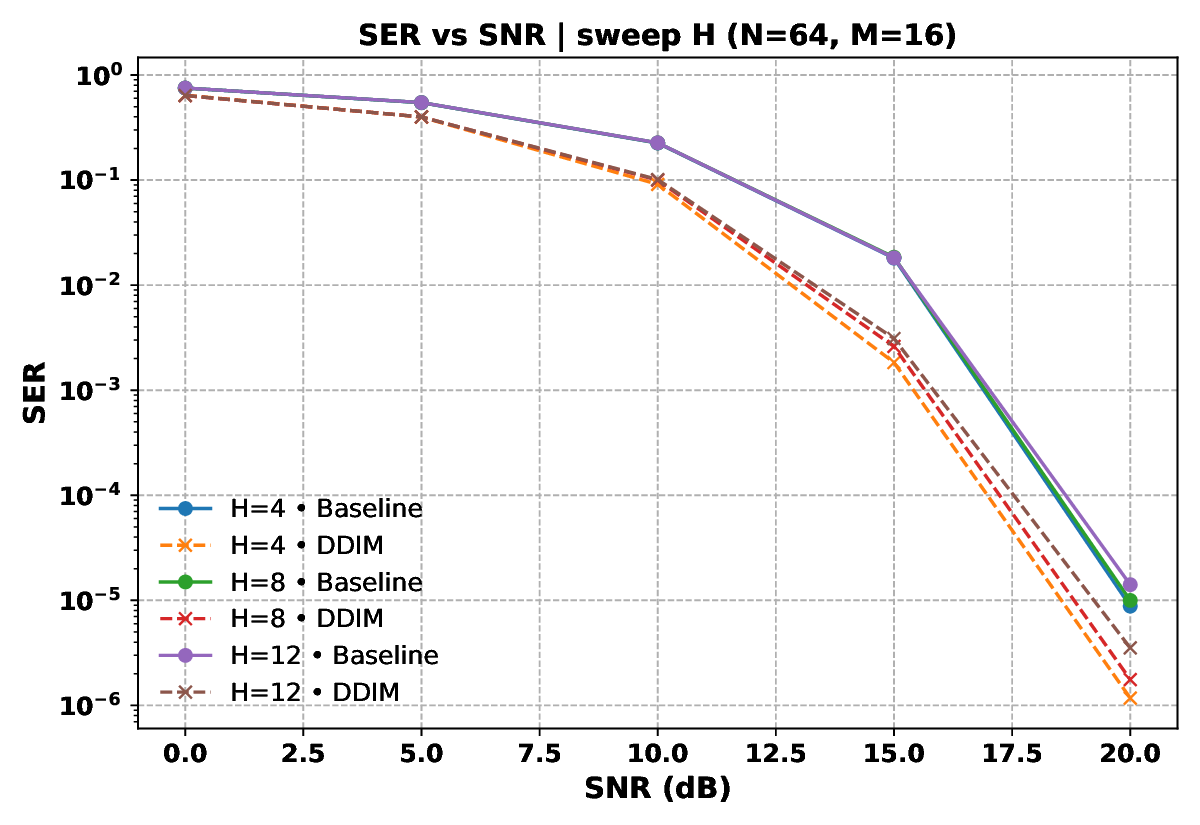}%
        \label{fig:ser_n64_m16}}
    \hfill
    \subfloat[BER]{%
        \includegraphics[width=0.32\textwidth]{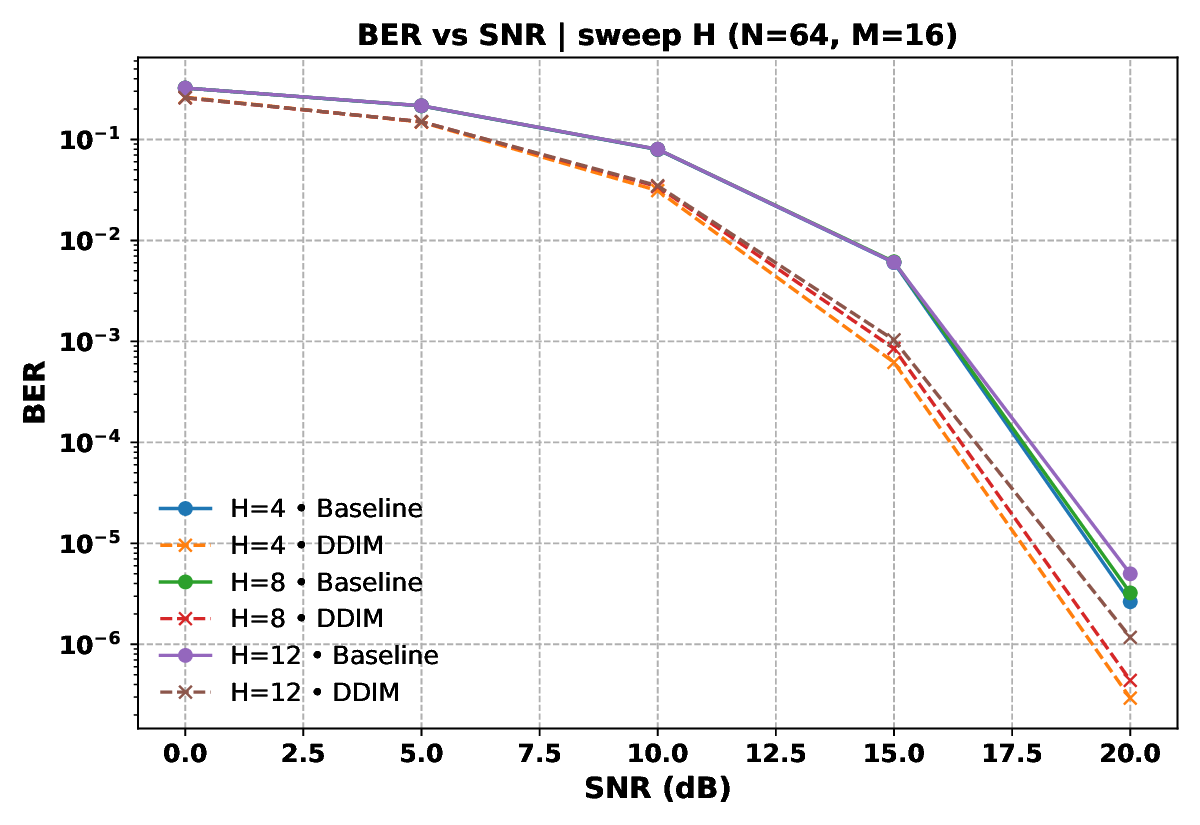}%
        \label{fig:ber_n64_m16}}
    \caption{Performance under AWGN-only propagation as a function of hop count $H$ for $M=16$ and $N=64$: 
(a) MSE, (b) SER, and (c) BER. 
Results confirm that AF--DDIM effectively suppresses the error growth caused by multi-hop noise accumulation.}
    \label{fig:perf_vs_H_n64_m16}
\end{figure*}

\subsection{Results}
We assess the performance of our proposed approach in two distinct propagation regimes. First, we examine the AWGN-only case, in which fading effects are absent and each relay hop introduces only additive Gaussian noise. This configuration isolates the fundamental behavior of AF–DDIM regarding noise accumulation in the absence of channel-state signaling.
Next, we consider a more practical scenario characterized by Rician fading ($K=15$dB), random hop distances, and significant large-scale path loss. In this setting, the sufficient statistics $(\mu_H, v_H)$ must be quantized and forwarded, allowing us to evaluate both the denoising performance of AF–DDIM under fading conditions and the impact of the low-complexity CSI model detailed in Section\ref{sec:implementation}.
Unless noted otherwise, we use a block length of $N^2=4096$ symbols, ten hops ($H=10$), and modulation orders $M\in{4,16,64}$ (square QAM). Source symbols are normalized to unit average power prior to transmission. For the fading regime, hop distances are sampled uniformly from $[1,2]$~m with a path-loss exponent of $2$ and a reference attenuation of $10$~dB at $1$~m. Independent realizations of fading and noise are generated for each Monte Carlo trial.

Training data are synthesized using the forward AF--DDIM process: for each sample, a random hop index $t\in\{1,\ldots,H\}$ is selected, 
a constellation symbol is propagated through $t$ AF stages with randomized fading, and the resulting noisy block and cumulative statistics are stored. 
The denoiser is trained offline with 10,000 samples for each configuration and validated using 400 samples. Training employs a batch size of 64 and runs for 5 epochs, utilizing the AdamW optimizer with a learning rate of $10^{-3}$. Performance is assessed on 400 independent test realizations per configuration.
We evaluate performance using mean-squared error (MSE), symbol error rate (SER), and bit error rate (BER), each analyzed as a function of SNR. Additionally, we assess the influence of block size ($N$), hop count ($H$), and modulation order ($M$) to demonstrate the robustness of AF–DDIM across different propagation depths and signaling overheads.

For comparison, we benchmark against a conventional AF detection pipeline that does not use the diffusion formulation. 
At the destination, the received signal is modeled as 
\[
Y = \bar\alpha_H X_0 + \sqrt{1-\bar\alpha_H^2}\, Z,
\qquad Z \sim \mathcal{CN}(0,I),
\]
and symbol-wise detection is performed using a classical maximum-likelihood (ML). 
For coded transmissions, the baseline employs a standard belief-propagation (BP) decoder operating on log-likelihood ratios derived from this effective channel. 
Unlike the proposed AF--DDIM method, the baseline does not explicitly track the sufficient statistics $(\mu_H,v_H)$ across hops, but instead relies only on the end-to-end effective SNR. 
This represents a natural and widely accepted reference design: it is near-optimal within the classical AF framework, yet does not leverage generative denoising priors or hop-wise statistics.

Figure~\ref{fig:perf_vs_M_n64_h10} illustrates performance under AWGN-only propagation for a block size of $N=64$ symbols and $H=10$ hops, comparing baseline direct detection with AF--DDIM reconstruction across different modulation orders. In Fig.~\ref{fig:mse_n64_h10}, the mean-squared error (MSE) is consistently reduced by AF--DDIM, with the most pronounced improvement at moderate SNR values (5--15~dB), where relay noise dominates. At high SNR, both approaches converge as the additive noise vanishes. In Fig.~\ref{fig:ser_n64_h10}, the symbol error rate (SER) decreases with SNR as expected, and AF--DDIM achieves substantially lower error rates, particularly for higher-order constellations ($M=16,64$). Figure~\ref{fig:ber_n64_h10} shows the corresponding bit error rate (BER) trends, which mirror the SER behavior: AF--DDIM preserves reliable decoding at SNR values where the baseline detector fails.

Figure~\ref{fig:perf_vs_M_n64_h10_real} considers a Rician fading environment ($K=15$~dB) with random hop distances and path-loss exponent $2$. 
For $N=64$ and $H=10$, AF--DDIM continues to provide a clear advantage across modulation orders. 
In Fig.~\ref{fig:mse_n64_h10_real}, MSE reduction is consistent despite fading-induced variability, 
and in Figs.~\ref{fig:ser_n64_h10_real}--\ref{fig:ber_n64_h10_real}, both SER and BER improvements are preserved for $M=16$ and $M=64$, 
highlighting that the matched diffusion schedule adapts well to fading statistics. 
The performance gap is most visible at moderate SNR, where fading and AWGN jointly degrade the baseline detector.

\begin{figure*}[!t]
    \centering
    \subfloat[MSE]{%
        \includegraphics[width=0.32\textwidth]{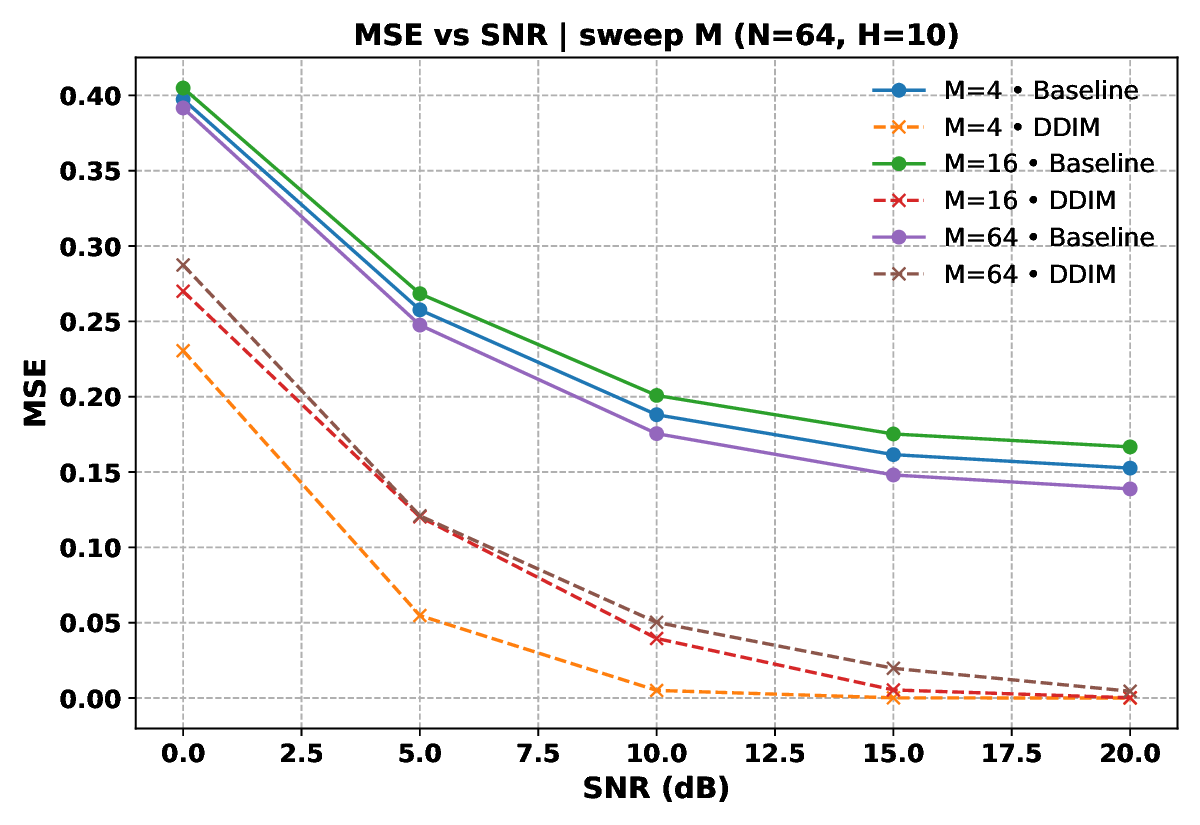}%
        \label{fig:mse_n64_h10_real}}
    \hfill
    \subfloat[SER]{%
        \includegraphics[width=0.32\textwidth]{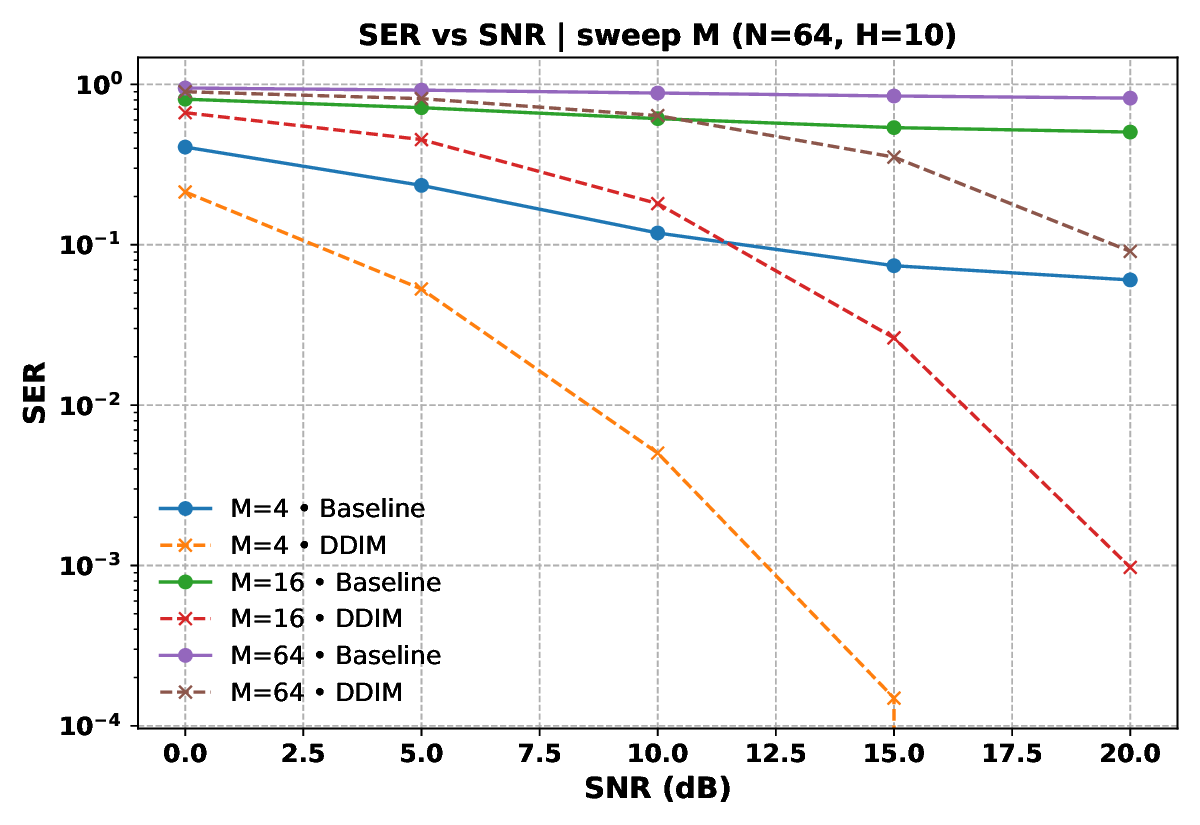}%
        \label{fig:ser_n64_h10_real}}
    \hfill
    \subfloat[BER]{%
        \includegraphics[width=0.32\textwidth]{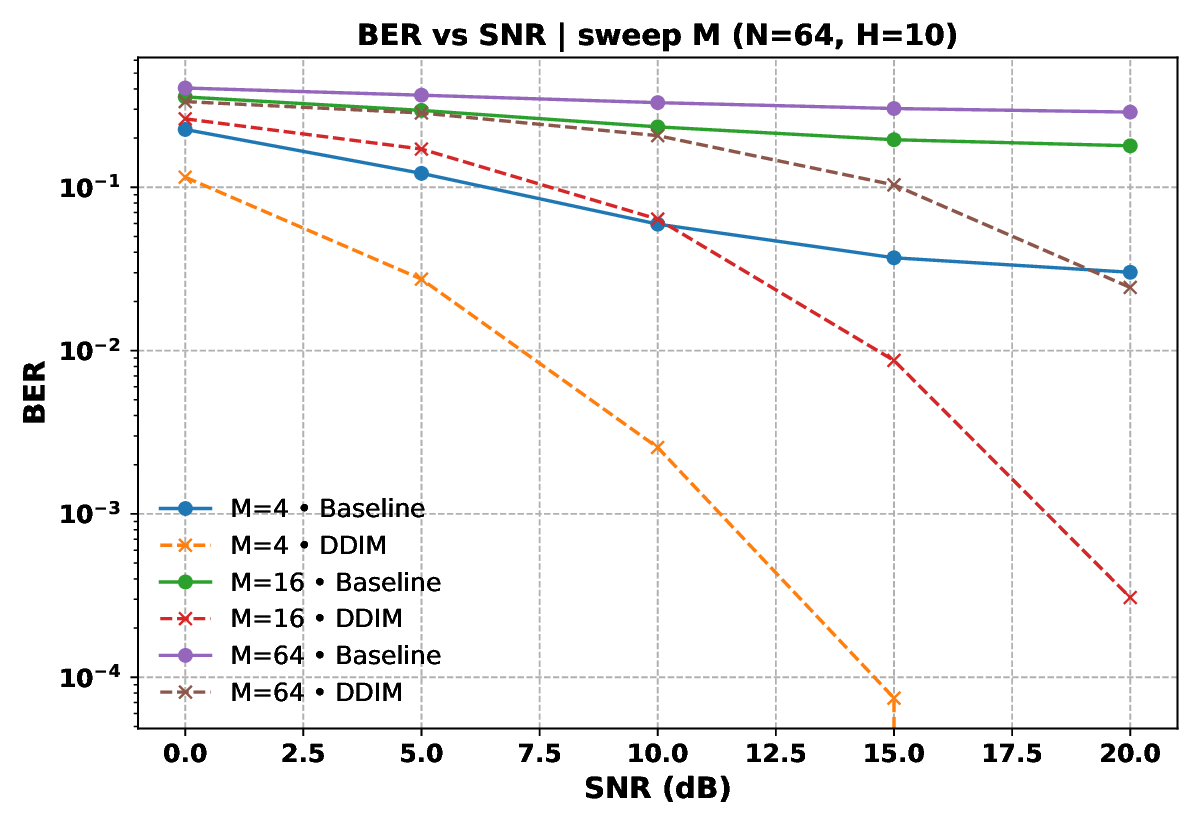}%
        \label{fig:ber_n64_h10_real}}
    \caption{Performance under Rician fading ($K=15$~dB, path-loss exponent 2, hop distances uniformly distributed in $[1,2]$\,m, reference path loss $10$~dB at $1$\,m) 
versus modulation order $M$ for $N=64$ and $H=10$: 
(a) MSE, (b) SER, and (c) BER. 
AF--DDIM achieves consistent error reduction across modulation orders despite fading variability.}
    \label{fig:perf_vs_M_n64_h10_real}
\end{figure*}

Figure~\ref{fig:perf_vs_N_m16_h10_real} studies the role of block size $N$ under the same Rician fading model with $M=16$ and $H=10$. 
Increasing $N$ improves utilization by amortizing CSI overhead, which translates into better denoising performance. 
AF--DDIM consistently achieves lower MSE (Fig.~\ref{fig:mse_m16_h10_real}) and reduced SER/BER (Figs.~\ref{fig:ser_m16_h10_real}--\ref{fig:ber_m16_h10_real}) across all $N$, 
showing that the framework remains effective even when fading is present. 
These results validate the low-overhead design, as reliable performance is reached with block sizes as small as $N \approx 20$.

\begin{figure*}[!t]
    \centering
    \subfloat[MSE]{%
        \includegraphics[width=0.32\textwidth]{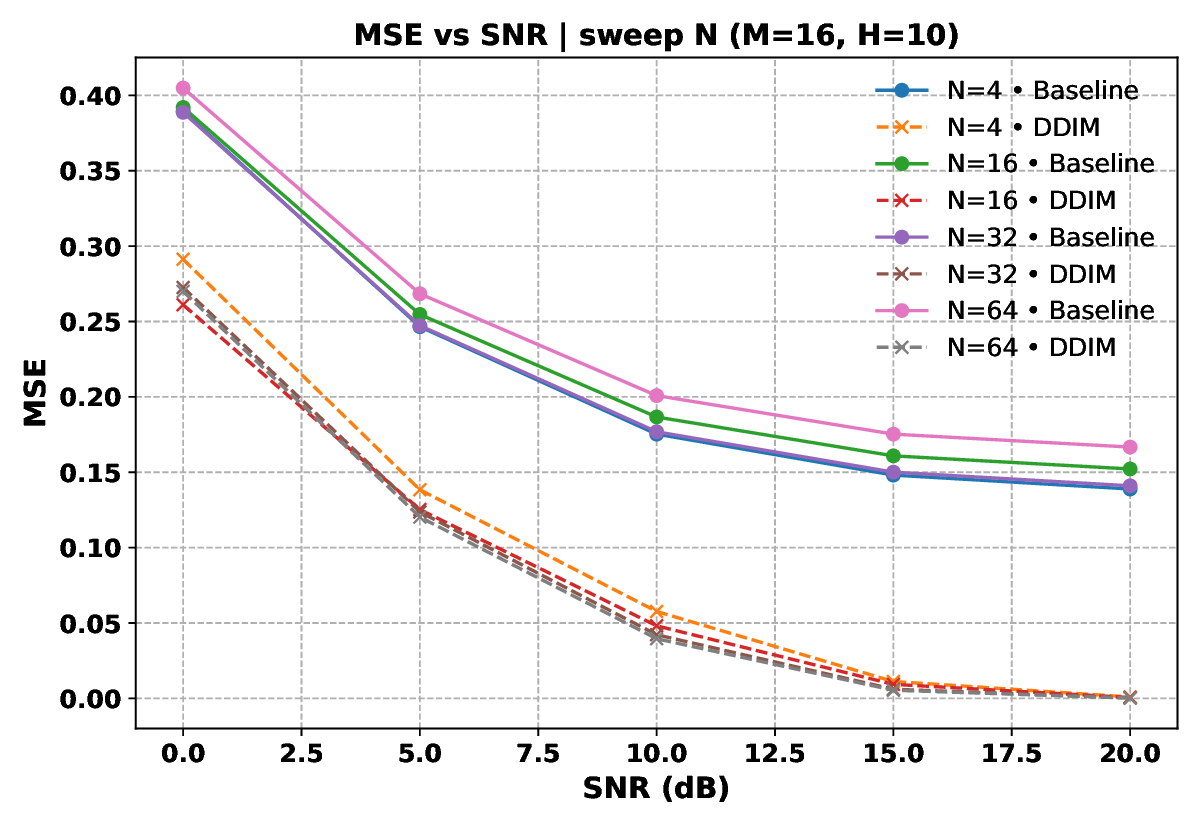}%
        \label{fig:mse_m16_h10_real}}
    \hfill
    \subfloat[SER]{%
        \includegraphics[width=0.32\textwidth]{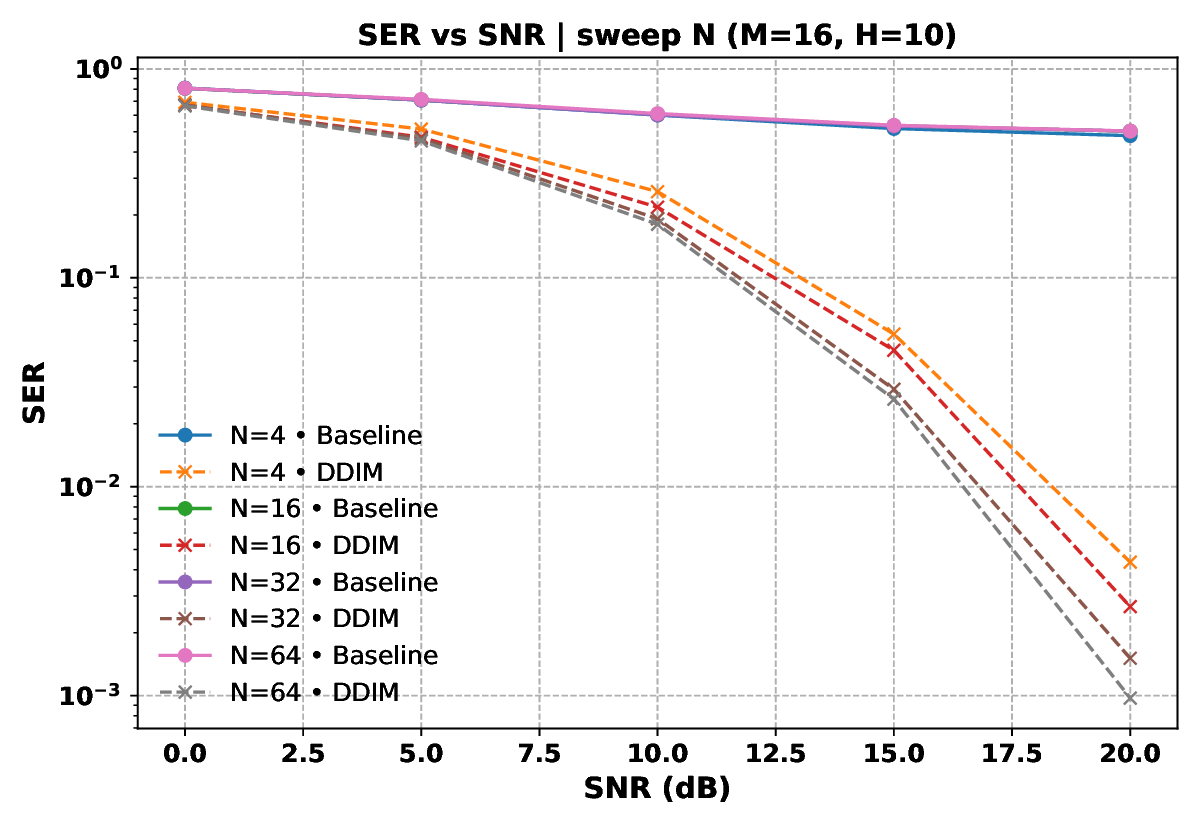}%
        \label{fig:ser_m16_h10_real}}
    \hfill
    \subfloat[BER]{%
        \includegraphics[width=0.32\textwidth]{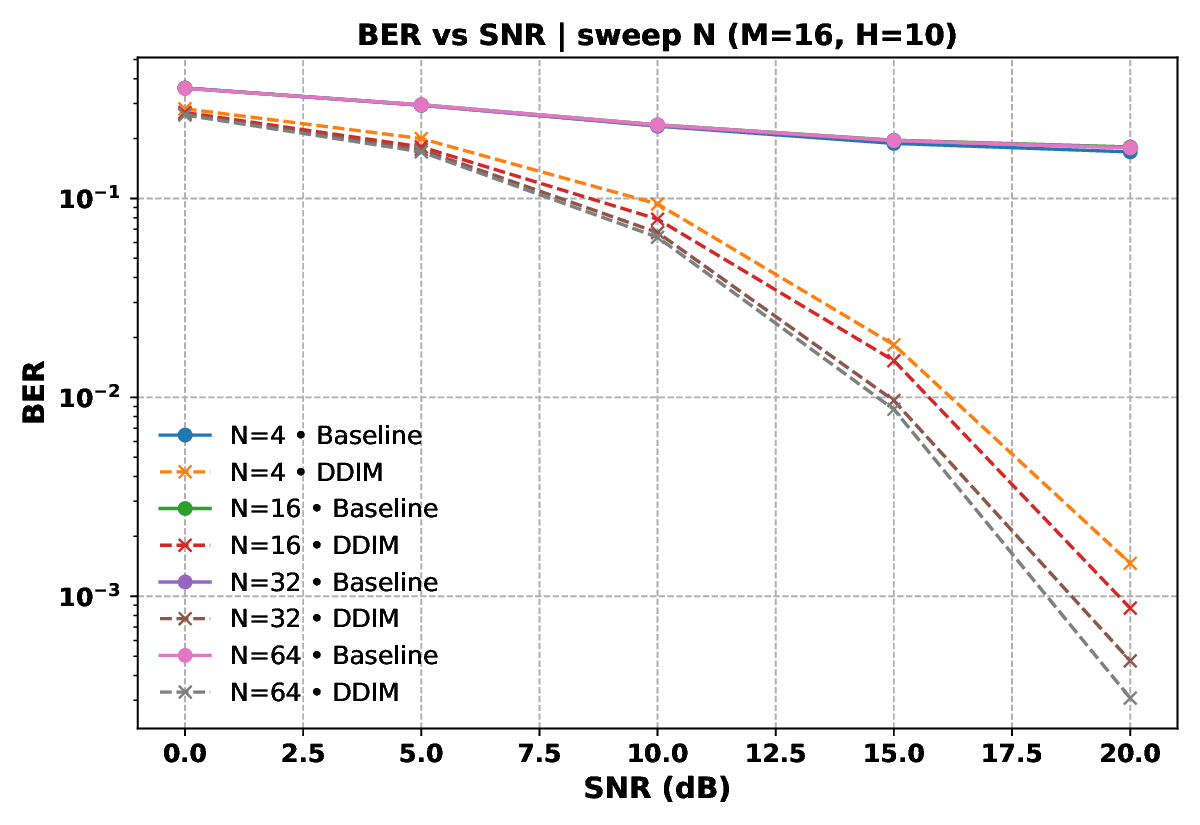}%
        \label{fig:ber_m16_h10_real}}
    \caption{Performance under Rician fading ($K=15$~dB, path-loss exponent 2, hop distances uniformly distributed in $[1,2]$\,m, reference path loss $10$~dB at $1$\,m) 
as a function of block size $N$ for $M=16$ and $H=10$: 
(a) MSE, (b) SER, and (c) BER. 
Larger $N$ improves utilization and enhances AF--DDIM performance, while clear denoising gains persist even for moderate block sizes.}
    \label{fig:perf_vs_N_m16_h10_real}
\end{figure*}

Figure~\ref{fig:perf_vs_H_n64_m16_real} evaluates sensitivity to hop count under Rician fading with $M=16$ and $N=64$. 
Baseline detection suffers severe degradation as $H$ increases, due to both noise amplification and fading compounding across hops. 
AF--DDIM substantially mitigates this effect: in Fig.~\ref{fig:mse_n64_m16_real}, the MSE remains well controlled, 
and in Figs.~\ref{fig:ser_n64_m16_real}--\ref{fig:ber_n64_m16_real}, SER and BER remain markedly lower than baseline across all $H$. 
These results confirm that the proposed framework scales gracefully with network depth, even under realistic fading propagation.

\begin{figure*}[!t]
    \centering
    \subfloat[MSE]{%
        \includegraphics[width=0.32\textwidth]{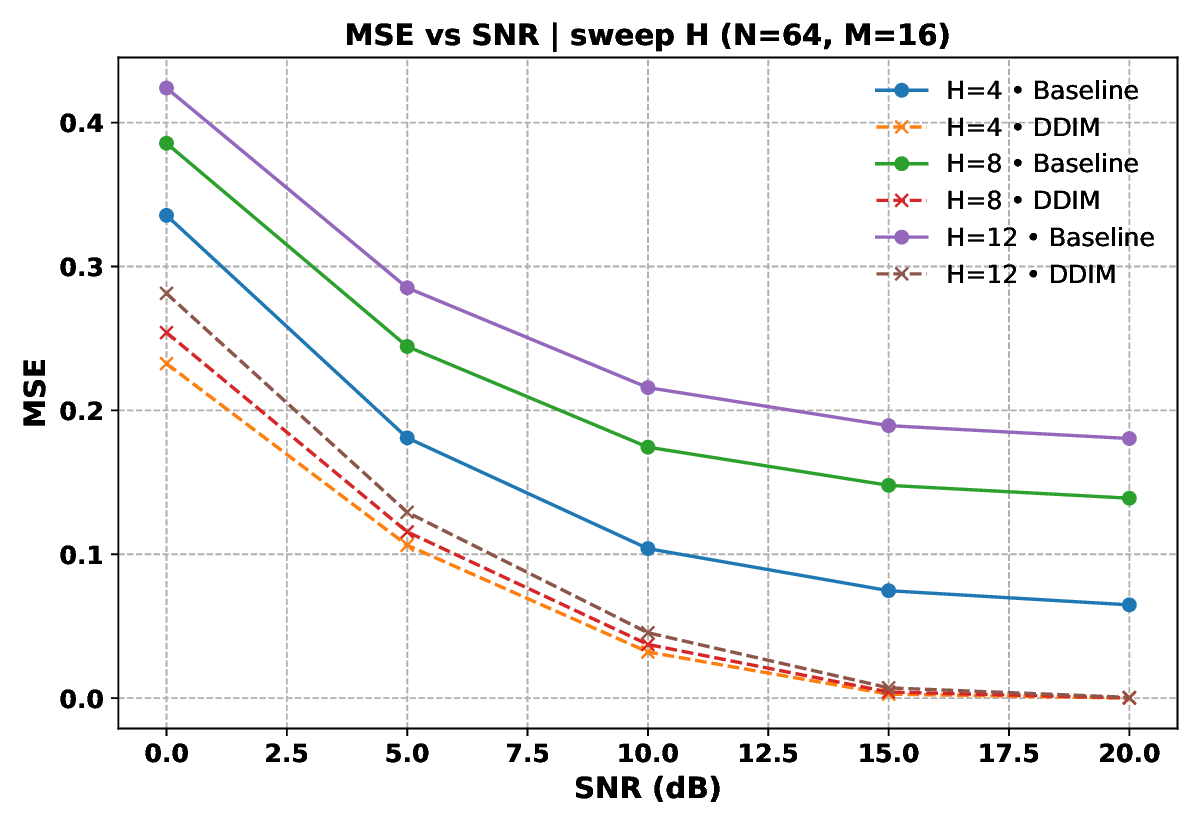}%
        \label{fig:mse_n64_m16_real}}
    \hfill
    \subfloat[SER]{%
        \includegraphics[width=0.32\textwidth]{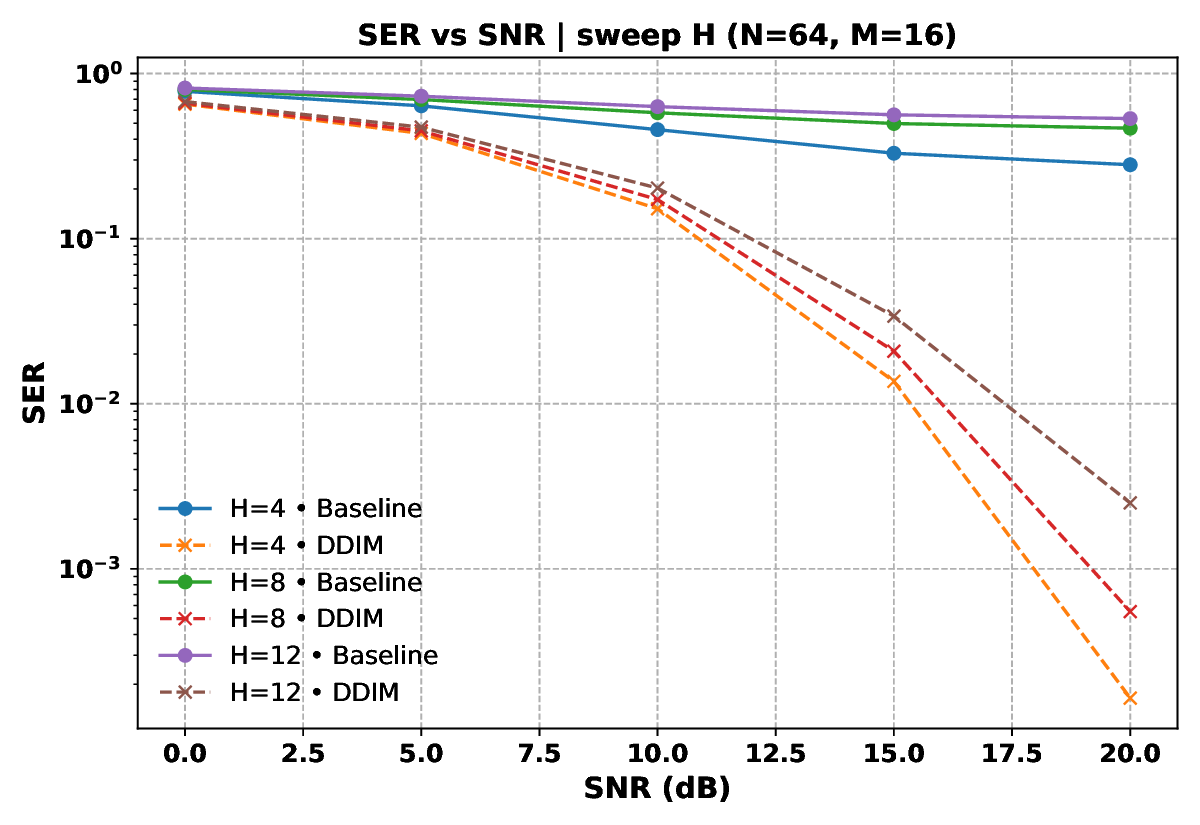}%
        \label{fig:ser_n64_m16_real}}
    \hfill
    \subfloat[BER]{%
        \includegraphics[width=0.32\textwidth]{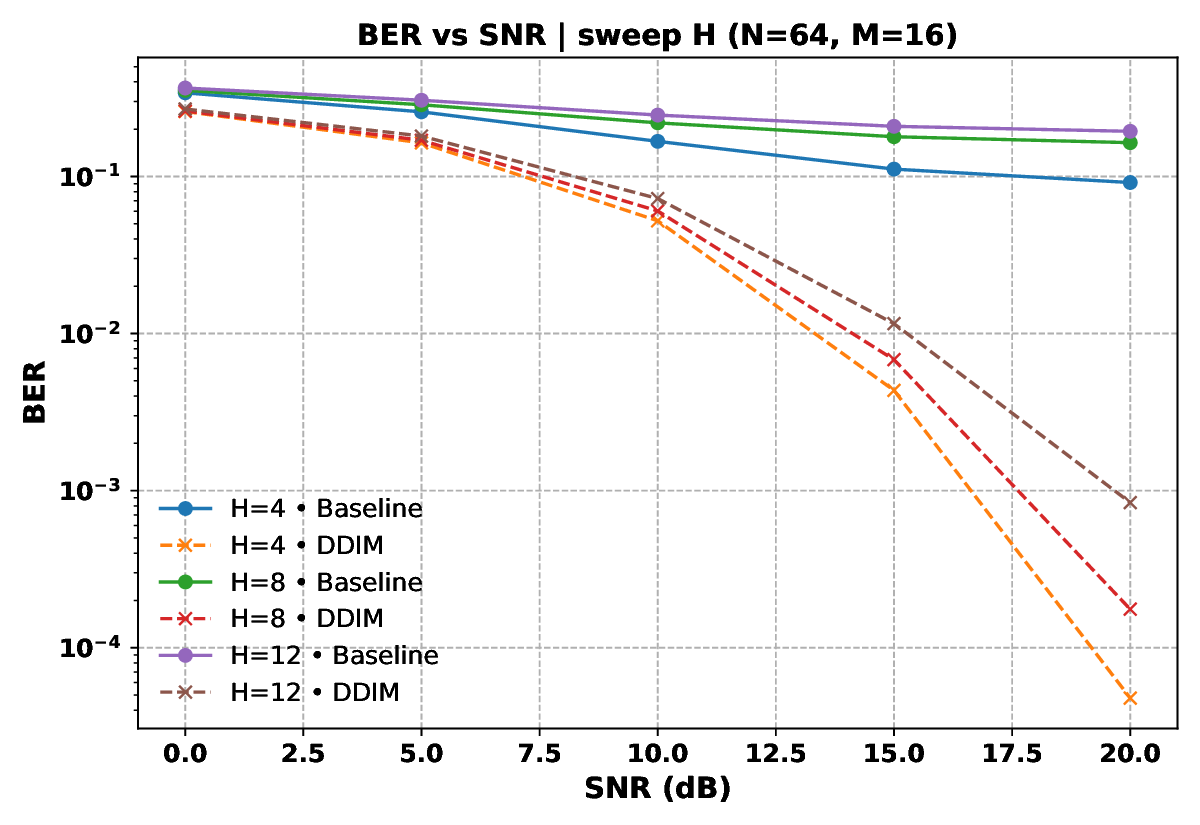}%
        \label{fig:ber_n64_m16_real}}
    \caption{Performance under Rician fading ($K=15$~dB, path-loss exponent 2, hop distances uniformly distributed in $[1,2]$\,m, reference path loss $10$~dB at $1$\,m) 
as a function of hop count $H$ for $M=16$ and $N=64$: 
(a) MSE, (b) SER, and (c) BER. 
AF--DDIM robustly controls error growth with increasing network depth, demonstrating resilience to fading and multi-hop noise accumulation.}
    \label{fig:perf_vs_H_n64_m16_real}
\end{figure*}

\section{Conclusions}
We have proposed a framework that reinterprets multi-hop amplify-and-forward relaying as a physical instantiation of a forward diffusion process. By collapsing per-hop channel effects into end-to-end sufficient statistics, the framework enables a low-overhead signaling model where only three scalars are forwarded per block. At the receiver, these statistics parameterize a matched reverse-time schedule that allows a DDIM-based denoiser to reconstruct the transmitted signal. 

Extensive numerical evaluations under AWGN and Rician fading channels confirm the effectiveness of AF--DDIM in mitigating multi-hop noise accumulation. The framework achieves substantial gains in MSE, SER, and BER compared to direct detection, with improvements most pronounced at moderate SNR values and for higher-order constellations. Importantly, the CSI overhead is shown to be negligible: even conservative quantization yields utilization above 95\% for modest block sizes, validating the efficiency of the proposed approach.

These findings highlight our approach as a scalable and physically interpretable framework for reliable sub-THz multi-hop communications. Future work may extend the methodology to frequency-selective channels, adaptive quantization of sufficient statistics, and joint training of denoisers across heterogeneous relay networks.

\balance

\bibliographystyle{IEEEtran} 
\bibliography{IEEEabrv, ref, references}
\end{document}